\documentclass[letterpaper,twocolumn,10pt]{article} 
\usepackage{usenix2019_v3}
\usepackage{url}
\usepackage{adjustbox}
\usepackage{amssymb}
\usepackage{amsmath}
\usepackage{amsthm}
\usepackage{anyfontsize}
\usepackage{booktabs}
\usepackage[font={small}]{caption}
\usepackage{color}
\usepackage{colortbl}
\usepackage{enumitem}
\usepackage{framed}
\PassOptionsToPackage{hyphens}{url}
\usepackage{letltxmacro}
\usepackage{mathtools}
\usepackage{microtype}
\usepackage{pgfplotstable}
\usepackage{pgfplots}
\usepackage{scalefnt}
\usepackage{setspace}
\usepackage{tikz}
\usetikzlibrary{calc}
\usepackage{xcolor}
\usepackage{xspace}
\usepackage{titlesec}
\usepackage{wrapfig}
\usepackage{subcaption}
\usepackage[most]{tcolorbox}
\usepackage{lipsum}

\newtcolorbox{myquote}[1][]{%
    colback=black!5,
    colframe=black!5,
    notitle,
    sharp corners,
    borderline west={2pt}{0pt}{blue!80!black},
    enhanced,
    breakable,
    }

\newcommand{\F}{\mathbb{F}}

\theoremstyle{plain}
\newtheorem{theorem}{Theorem}

\newtheorem*{claim}{Claim}
\theoremstyle{definition}

\newtheorem{definition}[theorem]{Definition}

\newtheoremstyle{goal}
  {\topsep}
  {\topsep}
  {\normalfont}
  {0pt}
  {\bfseries}
  {: } 
  { }
  {\thmname{#1}\thmnumber{ #2}\thmnote{ (#3)}}
\theoremstyle{goal}

\newcounter{ExperimentCount}

\titleformat*{\subsection}{\fontsize{11}{12} \selectfont\bfseries}
\titleformat*{\subsubsection}{\bfseries}
\titlespacing*{\subsection}{0pt}{1.75ex plus 0.5ex minus 0.5ex}{0.5ex plus .25ex minus .25ex}

\makeatletter
\g@addto@macro{\UrlBreaks}{\UrlOrds}
\makeatother

\usepackage[T1]{fontenc}
\usepackage{newtxmath}
\usepackage{newtxtext}

\usepackage[scaled=.8]{beramono}

\DeclareMathAlphabet\mathcal{OMS}{cmsy}{m}{n}
\SetMathAlphabet\mathcal{bold}{OMS}{cmsy}{b}{n}

\DeclareSymbolFont{AMSb}{U}{msb}{m}{n}
\DeclareSymbolFontAlphabet{\mathbb}{AMSb}

\DeclareSymbolFont{numbers}{T1}{ptm}{m}{n}
\SetSymbolFont{numbers}{bold}{T1}{ptm}{bx}{n}
\DeclareMathSymbol{0}\mathalpha{numbers}{"30}
\DeclareMathSymbol{1}\mathalpha{numbers}{"31}
\DeclareMathSymbol{2}\mathalpha{numbers}{"32}
\DeclareMathSymbol{3}\mathalpha{numbers}{"33}
\DeclareMathSymbol{4}\mathalpha{numbers}{"34}
\DeclareMathSymbol{5}\mathalpha{numbers}{"35}
\DeclareMathSymbol{6}\mathalpha{numbers}{"36}
\DeclareMathSymbol{7}\mathalpha{numbers}{"37}
\DeclareMathSymbol{8}\mathalpha{numbers}{"38}
\DeclareMathSymbol{9}\mathalpha{numbers}{"39}

\SetSymbolFont{numbers}{bold}{T1}{ptm}{b}{n}
\makeatletter
\renewcommand{\operator@font}{\mathgroup\symnumbers}
\makeatother

\definecolor{LightCyan}{rgb}{0.88,1,1}

\setitemize{noitemsep,topsep=2pt,parsep=2pt,partopsep=2pt}

\long\def\com#1{}

\renewcommand{\paragraph}[1]{\medskip\noindent\textbf{#1}}

\makeatletter

\let\c@table\c@figure
\makeatother 

\setlist[description]{leftmargin=\parindent,topsep=0ex,itemsep=0ex,partopsep=1ex,parsep=1ex}

\LetLtxMacro{\oldtextsc}{\textsc}
\renewcommand{\textsc}[1]{\oldtextsc{\scalefont{1.2}#1}}

\newcommand{\ignore}[1]{}

\newcommand{\name}{Express\xspace}

\newenvironment{denseitemize}{
\begin{itemize}[topsep=2pt, partopsep=0pt, leftmargin=1.5em]
  \setlength{\itemsep}{4pt}
  \setlength{\parskip}{0pt}
  \setlength{\parsep}{0pt}
}{\end{itemize}}

\newenvironment{denseenumerate}{
\begin{enumerate}[topsep=2pt, partopsep=0pt, leftmargin=1.5em]
  \setlength{\itemsep}{4pt}
  \setlength{\parskip}{0pt}
  \setlength{\parsep}{0pt}
}{\end{enumerate}}

\newcommand{\rev}[1]{#1}

\begin{document} 
\frenchspacing


\date{}
\title{Express: Lowering the Cost of Metadata-hiding \\Communication with Cryptographic Privacy}
\author{
{\rm Saba Eskandarian}\\
Stanford University 
\and
{\rm Henry Corrigan-Gibbs}\\
MIT CSAIL 
\and 
{\rm Matei Zaharia}\\
Stanford University
\and 
{\rm Dan Boneh}\\
Stanford University
}

\maketitle

\begin{abstract}
Existing systems for metadata-hiding messaging that provide cryptographic privacy properties have either high communication costs, high computation costs, or both. In this paper, we introduce \name, a metadata-hiding communication system that significantly reduces both communication and computation costs. 
\name is a two-server system that provides cryptographic security against an arbitrary number of malicious clients and one malicious server. 
In terms of communication, \name only incurs a constant-factor overhead per message sent regardless of the number of users, whereas previous cryptographically-secure systems Pung and Riposte had communication costs proportional to roughly the square root of the number of users. 
In terms of computation, \name only uses symmetric key cryptographic primitives and makes both practical and asymptotic improvements on protocols employed by prior work. 
These improvements enable \name to increase message throughput, reduce latency, and consume over $100\times$ less bandwidth than Pung and Riposte, dropping the end to end cost of running a realistic whistleblowing application by~$6\times$. 
\end{abstract}

\section{Introduction}

Secure messaging apps and TLS protect the confidentiality of data in transit.
However, transport-layer encryption does little to protect sensitive communications \textit{metadata},
which can include the time of a communications session, the identities of the
communicating parties, and the amount of data exchanged.
As a result, state-sponsored intelligence gathering and surveillance programs~\cite{metadata,article1,article2}, particularly those targeted at journalists and dissidents~\cite{unesco,UN}, continue to thrive \rev{-- even in strong democracies like the United States~\cite{article3,article4}}. Anonymity systems such as Tor~\cite{tor}, or the whistleblowing tool SecureDrop~\cite{securedrop,securedropreport}, attempt to hide communications metadata, but they are vulnerable to traffic analysis if an adversary controls certain key points in the network~\cite{breaktor,HB13,JWJ+13}.

A host of systems can hide metadata 
with cryptographic security guarantees (e.g., Riposte~\cite{riposte}, Talek~\cite{talek}, P3~\cite{P3}, Pung~\cite{pung}, Riffle~\cite{riffle}, Atom~\cite{atom}, XRD~\cite{xrd}).
Unfortunately, these systems generally use heavy public-key cryptographic tools and incur high communication costs, making them difficult to deploy in practice. 
Another class of systems provides a differential privacy security guarantee (e.g., Vuvuzela~\cite{vuvuzela}, Alpenhorn~\cite{alpenhorn}, Stadium~\cite{stadium}, Karaoke~\cite{karaoke}).
These systems offer high throughput and very low communication costs, but their security guarantees degrade with each round of communication, making
them unsuitable for communication infrastructure that must operate over a long period of time.

This paper presents \name, a metadata-hiding communication system with cryptographic security 
that makes both practical and asymptotic improvements over prior work. 
\name is a two-server system that provides
cryptographic security against an arbitrary number
of malicious clients and up to one malicious
server. This security guarantee falls between that of Riposte~\cite{riposte}, which provides security against at most one malicious server out of three total, and Pung~\cite{pung}, which can provide security even in the single-server setting where the server is malicious. \name only uses lightweight symmetric cryptographic primitives and introduces new protocols which allow it to improve throughput, reduce latency, consume over $100\times$ less bandwidth, and cost $6\times$ less to operate compared to these prior works. 

\paragraph{\name architecture.} To receive messages via \name, a client registers \emph{\rev{mailboxes}} with the servers,
who collectively maintain the contents of all the mailboxes.
After registration, the mailbox owner distributes the 
\emph{address} of a mailbox (i.e., a cryptographic identifier)
to \rev{each} communication peer via some out-of-band means. 
Given the address of a mailbox, any client
can use \name to upload a message into
that mailbox, without revealing to anyone except the mailbox owner which
mailbox the client wrote into.
Mailbox owners can fetch the contents of their
mailboxes at any time with any frequency they wish,
and \emph{only} the owner of a mailbox can fetch
its contents.

Crucially, \name hides which client wrote into which mailbox 
but \emph{does not} hide which client \emph{read} from which mailbox. This requires mailbox owners to check their mailboxes at a fixed frequency, although there need not be any synchronization between the rates that different owners access their mailboxes. 
As we will discuss, this form of metadata privacy fits well with our main application: whistleblowing. 

\paragraph{Technical overview.} We now sketch the technical ideas behind the design of \name. 
As in prior work~\cite{riposte},
\name servers hold a table of mailboxes 
secret-shared across two servers; 
clients use a cryptographic tool called a 
\emph{distributed point function}~\cite{dpfs} to 
write messages into a mailbox without 
the servers learning which mailbox a
client wrote into~\cite{OS97,riposte}. 
This basic approach to private writing leaves two important problems unsolved: handling read access to mailboxes and dealing with denial of service attacks from malicious users. 

The first contribution of \name is to allow 
mailbox reads and writes to be asynchronous. 
This allows \name clients to contact the system with any frequency they like, regardless of other clients' behavior. 
In contrast, prior systems such as Riposte, Pung, and Vuvuzela~\cite{riposte, pung, vuvuzela} require \emph{every} client to write before \emph{any} client can read, so the whole system is forced to operate in synchronized rounds. 
We are able to allow read/write interleaving in \name 
with a careful combination of encryption and rerandomization. 
At a high level: any client in \name can read from any mailbox,
but each read returns a fresh re-randomized encryption of the mailbox contents that
only the mailbox owner can decrypt.
In this way, even if an adversary reads the contents of
all mailboxes between every pair of client writes, the adversary
learns nothing about which honest client is communicating with which
honest client.

The second major challenge for messaging systems based on 
secret sharing~\cite{Chaum88,dissent,dissentnumbers,verdict,riposte,prio}
is to protect against malicious clients, who may corrupt the functioning
of the system by submitting malformed messages.
Since no server has a complete view of the message being written by each client, servers cannot immediately tell if a message is well-formed, e.g., whether it modifies only one mailbox or overwrites the contents of many mailboxes with garbage, destroying real messages that may have been placed in them. 
\name protects against such denial-of-service attacks using
a new \emph{auditing} protocol.
In a system with $n$ mailboxes,
\name's auditing protocol requires only 
$O(\lambda)$ communication between parties, for a fixed security parameter $\lambda$,
as well as $O(1)$ client side 
computation (in terms of AES evaluations and finite field operations). 
The analogous scheme in Riposte required $\Omega(\lambda\sqrt n)$ 
communication and $\Omega(\sqrt n)$ client computation~\cite{riposte}, and additionally
required a third non-colluding server. 
In practice, our new auditing scheme reduces overall computation costs for the client by
$8\times$ for a deployment with one million registered mailboxes. 

In addition to defending against malformed messages aimed at corrupting the whole database of mailboxes, \name must protect against targeted attacks. A malicious client could potentially send a correctly-formed message containing random content to a single mailbox in hopes of overwriting any content written to that mailbox by an honest client. We defend against this by assigning \emph{virtual addresses} to each mailbox. Each mailbox is accessed via a 128-bit virtual address, regardless of the actual number of mailboxes registered. The servers store and compute only over the number of actually registered mailboxes, not the number of virtual mailboxes. However, since virtual addresses are distributed at random over an exponentially large address space, a malicious client cannot write to a mailbox unless it knows the corresponding address. Section~\ref{sec:malicious} describes our protections against malicious clients in detail. 

\paragraph{Evaluation application.} 
We evaluate \name as a system for \rev{whistleblowers to send messages to journalists}
while hiding their communications metadata from network surveillance.
In this application, a journalist registers a mailbox for each source from
which she wishes to receive information.
The journalist then communicates her mailbox address to the source via,
for example, a one-time in-person meeting. 
Thereafter, the source can privately send messages to the journalist by 
dropping them off in the journalist's \name mailbox.
In this way, we can implement a cryptographically metadata-hiding
variant of the SecureDrop system~\cite{securedropreport}.

To provide whistleblowers with any reasonable guarantee of privacy, \name must provide its users with a degree of plausible deniability in the form of cover traffic. Otherwise, merely contacting the \name servers would automatically incriminate clients. 
As we will demonstrate, \name's low client computation and communication costs mean that an \name client implemented in JavaScript and embedded in a web page can generate copious cover traffic. Browsers that visit a cooperative news site's home page can opt-in to generate cover traffic for the system by running a JavaScript client in the backgound -- thereby increasing the anonymity set enjoyed by clients using \name to whistleblow -- without negatively impacting end-users' web browsing experience. 
We discuss this and other considerations involved in using \name for whistleblowing, e.g., how a journalist can communicate a mailbox address to a source, in Section~\ref{sec:dd}.

We implement \name and evaluate its performance on message sizes of up to 32KB, larger than is used in the evaluations of Pung~\cite{pung}, Riposte~\cite{riposte} and Vuvuzela~\cite{vuvuzela}. Recent high-profile whistleblowing events such as the whistleblower's report to the US intelligence community's inspector general~\cite{whistleblowingdoc} (25.3KB) or last year's anonymous New York Times op-ed~\cite{oped} (9KB) demonstrate that messages of this length are very relevant to the whistleblowing scenario. We also compare \name's performance to Pung~\cite{pung} and Riposte~\cite{riposte}, finding that \name matches or exceeds their performance, and conclude that \name reduces the dollar cost of running a metadata-hiding whistleblowing service by $6\times$ compared to prior work (see Figure~\ref{fig:costgraph}). On the client side, \name's computation and communication cost are both \emph{independent} of the number of users, at about $20$ms client computation and $5$KB communication overhead per message, enabling our new strategies for efficiently generating cover traffic. This represents over $100\times$ bandwidth savings compared to Riposte~\cite{riposte} and over $7,000\times$ savings compared to Pung for one million users. Although Vuvuzela operates under a very different security model, we compare the two systems qualitatively in our full evaluation, which appears in Section~\ref{eval}. 

In summary, we make the following contributions:
\begin{itemize}
\item The design and security analysis of \name, a metadata-hiding communication system that significantly reduces both communication and computation costs compared to prior work. 

\item A new auditing protocol to blindly detect malformed messages that is both asymptotically and practically more efficient than that of Riposte~\cite{riposte} while also removing the need for a third server to perform audits. 

\item An implementation and evaluation of \name that demonstrates the feasibility of our approach to metadata-hiding whistleblowing. Our open-source implementation of \name is available online at: \url{https://github.com/SabaEskandarian/Express}.
\end{itemize}

\section{Design Goals}\label{defs}

This section introduces the architecture of \name and describes our
security goals. 

An \name deployment consists of two servers that
collectively maintain a set of \textit{locked mailboxes}.
Each locked mailbox implements a private channel through which one client 
can send messages to another who has the 
secret cryptographic key to unlock that mailbox.

\ignore{
\begin{figure}
\centering
\includegraphics[width=.6\linewidth]{figures/newOverviewFig.pdf}
\caption{\small \name is a three-server system where two servers maintain clients' mailboxes, service read/write requests, and perform the bulk of computations. A third lightweight auditing server blindly monitors client writes to catch malformed requests that could corrupt other clients' mailbox contents.}
\label{overview}
\end{figure}
}

To use \name, a client wishing to receive 
messages first registers a mailbox and gets a \emph{mailbox address}. From then on, 
any client who has been given the mailbox address can write messages to that mailbox, 
and the owner of that mailbox can check the mailbox for messages whenever it wants. 
We discuss how clients can communicate mailbox addresses to each other via a \emph{dialing} protocol in Section~\ref{dialing}. 

We consider an attacker who controls one of the two \name servers, any number of \name
clients, and the entire network.
The main security property we demand is that, after an honest client \textit{writes} a message
into a mailbox, the attacker learns nothing about which mailbox the
client wrote into. 
\rev{This corresponds to an anonymity guarantee where the sender of a given message cannot be distinguished among the set of all senders in a given time interval.}
We also require that an attacker who controls any number of malicious clients
cannot prevent honest clients from communicating with each other.
In other words, we protect against malicious clients mounting 
in-protocol denial-of-service attacks. We do not aim to protect against DoS attacks by malicous servers, nor against network-level DoS attacks, \rev{but we will describe how clients can incorporate straightforward checks to detect tampering by malicious servers}.

\subsection{\name API}
\name allows clients to register mailboxes, read the contents of mailboxes they register, and privately write to others' mailboxes. Clients interact with the servers via the following
operations:

\smallskip
\noindent\textbf{Mailbox registration}. 
A client registers a new mailbox by sending the \name servers distinct \textit{mailbox keys}.
The servers respond with a \textit{mailbox address}. 
We say that a client ``owns'' a given mailbox if it holds the
mailbox's keys and address.

\smallskip
\noindent\textbf{Mailbox read}.
To read from a mailbox, the client sends the mailbox's address to the 
\name servers. 
The servers respond with the locked (i.e., encrypted) mailbox contents,
which the client can decrypt using its two mailbox keys together.

\smallskip
\noindent\textbf{Mailbox write}. 
To write to a mailbox, a client sends a specially-encoded \textit{write request} to the
\name servers that contains an encoding of both the address of the 
destination mailbox and the message to write into it.
No single \name server can learn either the 
destination address or message from the write request.

\subsection{Security Goals}\label{secgoals}

Based on the demands of our application to whistleblowing, \name primarily aims to provide privacy guarantees for \emph{writes} and not for reads. 
For example, \name hides who whistleblowers send messages to, but it does not hide the fact that journalists check their mailboxes. Below we describe \name's core security properties, which we formalize when proving security in Appendix~\ref{app:security}. 

\noindent\textbf{Metadata-hiding}. We wish to hide who a given client is writing to from everyone except the recipient of that client's messages. To this end, our \emph{metadata-hiding} security guarantee requires that for each write into an \name mailbox, no adversary who controls arbitrarily many clients and one server can determine which mailbox that write targeted unless the adversary owns the target mailbox.

We formalize this requirement in Appendix~\ref{app:security}, where we show
that an adversary can \emph{simulate} its view of honest clients' requests
before seeing them, which proves that the adversary learns nothing from
requests that it can't generate on its own, \rev{except necessary information such as the time the write occurred and which client initiated it}. In particular, this means
the adversary does not learn the mailbox into which \rev{a} request writes, 
\rev{although it does learn that a write has occurred}.
A malicious server can stop responding to requests or corrupt the
contents of users' mailboxes, but we require that even an actively malicious server
cannot break our metadata-hiding property.

\noindent\textbf{Soundness}. \name must be resilient to malformed messages sent by malicious clients. This means no client can write to a mailbox it has not been authorized to access, even if it deviates arbitrarily from our protocol. We capture this requirement via a \emph{soundness game} in Appendix~\ref{app:security}, where we also prove that no adversary can win the soundness game in \name with greater than negligible probability in a security parameter.

\subsection{\rev{Design Approaches}}\label{designapproaches}

\rev{As there are many potential approaches to metadata-hiding systems, we now briefly sketch high-level decisions made regarding the goals of \name.}

\paragraph{\rev{Deployment scenario.}} \rev{ \name's primary deployment scenario is as a system for whistleblowing, where a source leaks a document or tip to a journalist. In this setting, unlike prior work, \name does not require the system to run in synchronous rounds. This is the deployment scenario on which we will focus the exposition of the \name system. However, since this is a one-directional communication setting (the source can send leaks to the journalist but not have an ongoing conversation), \name can also be used as a standard messaging protocol where clients, e.g., sources and journalists, send messages back and forth to each other. In this setting, similar to prior work, messaging in \name would progress in time epochs, with a server-determined duration for each round. } 

\paragraph{\rev{Differential vs cryptographic privacy.}}
\rev{\name belongs to a family of systems that provide cryptographic security guarantees. In contrast, a number of systems (e.g., Vuvuzela, Stadium, Karaoke~\cite{vuvuzela,stadium,karaoke}) provide differentially private security. The difference between the two types of systems lies in the amount of private metadata the systems leak to an adversary. Cryptographic security means that no information leaks -- the adversary learns nothing, even after observing many rounds of communication, about which clients are communicating with each other. In contrast, systems providing the differential privacy notion of security allow some quantifiable leakage of metadata. Thus, with differential privacy-based systems, an attacker can -- after a number of communication rounds -- learn who is communicating. In contrast, the security of \name does not degrade, even after many rounds of interaction. Thus, although differentially private systems offer faster performance, cryptographic security is preferable for frequently used privacy-critical applications.}

\paragraph{\rev{Distributing trust.}} 
\rev{There are two potential approaches to deployment of metadata-hiding systems. One approach envisions a grass-roots deployment model where large numbers of people or organizations decide to participate to run the system, and trust is distributed among the servers with tolerance for some fraction behaving maliciously. The approach taken by \name (and the works to which we primarily compare it~\cite{riposte,pung}) envisions a commercial infrastructure setting where only a small number of participants (e.g., for our example use case, the Wall Street Journal and the Washington Post) are needed to deploy the system with its full security guarantees. Given equal performance and security against an equal fraction of malicious servers, it is of course preferable to distribute trust over a larger number of parties. Thus designs that split trust between a small number of parties can be seen as one point on a tradeoff between having many parties that undergo some light vetting versus having few parties that undergo heavier vetting before being included as servers in the system. }

\subsection{Limitations}\label{limitations}

We now discuss some limitations of \name to aid in determining which scenarios are best-suited to an \name deployment. 

\rev{The most important limitation to consider when deciding whether to deploy \name is the issue of censorship. As mentioned above, \name relies on distributing trust among two servers. Thus, if traffic to either server is blocked, the system can no longer be accessed. Since we envision \name being deployed by major news organizations, \name would not be appropriate for use in countries with a history of blocking traffic to such organizations. This is true of any system that distributes trust over a small number of servers (or has easily identifiable traffic). However, there is a need to prevent surveillance even in countries with relatively open access to the internet. It is in this setting that \name can be an effective approach to metadata-hiding communication.}

\name allows mailbox owners to access their mailboxes and retrieve messages with whatever frequency they desire \rev{when being used for one-way communication}, but they must check mailboxes at regular intervals in order to maintain security because \name does not hide which mailbox a given read accesses. If a mailbox owner changes her mailbox-checking pattern based on the contents of messages received, this may leak something about who is sending her messages. Note that although this implies that mailbox owners should regularly check their mailboxes, it does not impose any restrictions on the frequency with which any owner checks her mailboxes -- it is not a fixed frequency required by the system and can be different for each mailbox owner. This is in contrast with prior works, which fix a system-wide frequency with which clients must contact the servers or require clients to always remain online. Clients sending messages through \name but not also receiving messages (e.g., whistleblowers sending tips or documents) do not need to regularly contact the system.

Another reason for mailbox owners to check their mailboxes regularly is that messages in \name are written into mailboxes by adding, not concatenating, the message contents to the previous contents of the mailbox. It is thus possible for a second message sent to the same mailbox to overwrite the original contents, causing the content to be clobbered when someone eventually reads it. 
This risk can be easily mitigated, however, because each mailbox is for one client to send messages to one other client, and servers zero-out the contents of mailboxes after they are read to make space for new messages. 
Looking ahead to our application, messages can be a leak of a single document, where more than one message is not required. If a journalist expects to receive many messages from the same source before she has a chance to read and empty the contents of a mailbox, one way to handle this situation is to register several mailboxes for the same source, so each message can be sent to a different mailbox. This way, as long as a journalist checks and empties her mailboxes before they have all been used, no messages will be overwritten. 

\rev{While \name's soundness property prevents in-protocol denial of service attacks by malicious clients, a malicious \name server can launch a denial of service attack by overwriting mailboxes with garbage. This attack will prevent communication through \name, but it can at least be detected. We discuss how clients can add integrity checks to their messages to achieve authenticated encryption over \name in Section~\ref{sec:full}. This means that a client receiving a garbage message will know that the message has been corrupted by a malicious server.}

\rev{Finally, like all systems providing powerful metadata-hiding guarantees, \name must make use of cover traffic to hide information about which users are really communicating via \name. Although necessary, cover traffic allows metadata-hiding systems to protect even against adversaries with strong background knowledge about who might be communicating with whom by providing plausible deniability to clients sending messages through \name. We further discuss cover traffic in Section~\ref{covertraffic}.}

\section{\name Architecture} 
\label{sec:tech}

This section describes the basic architecture of \name. Section~\ref{sec:malicious} shows how to add defenses to protect against disruptive clients, and Section~\ref{sec:full} states the full \name protocol. 
Section~\ref{sec:dd} discusses how to use \name for whistleblowing,
including how a mailbox owner communicates a mailbox address to senders 
and how to increase the number of \name users by deploying it on the web.

The starting point for \name is a technique for privately writing into mailboxes using distributed point functions~\cite{dpfs,OS97,riposte}. We review how DPFs can be used for private writing in Section~\ref{reviewdpfs}. A private writing mechanism alone, however, 
does not suffice to allow metadata-hiding communication. 
We must also have a mechanism to handle access control so that only the mailbox owner 
can access the contents of a given mailbox. 
We discuss a lightweight cryptographic access control system in Section~\ref{accesscontrol}, where we also explain how this combination of private writing and controlled reading enables 
metadata hiding without synchronized rounds. 

\subsection{Review: Private Writing with DPFs}\label{reviewdpfs}
We briefly review the technique used in Riposte~\cite{riposte} for allowing
a client to privately write into a database, stored in secret-shared form, 
at a set of servers.

\noindent\textbf{A na\"ive approach}. In \name, two servers -- servers $A$ and $B$ -- collectively hold 
the contents of a set of \emph{mailboxes}.
In particular, if there are $n$ mailboxes in the system and each mailbox holds
an element of a finite field $\F$, then we can write the contents
of all mailboxes in the system as a vector $D \in \F^n$. 
Each server holds an additive secret share of the vector $D$:
that is, server $A$ holds a vector
$D_A \in \F^n$ and server $B$ holds a vector $D_B \in \F^n$ such
that $D = D_A + D_B \in \F^n.$

Once a client registers a mailbox, another client with that mailbox's address can send messages or documents to the mailbox, which the mailbox owner can check at his or her
convenience. Although \name can support mailboxes of different sizes, size information can be used to 
trace a message from its sender to its receiver, so \name clients must 
pad messages, either all to the same size or to one of a few pre-set size options. 

To write a message $m \in \F$ into the $i$-th mailbox na\"ively, the \name client
could prepare a vector $m \cdot \textbf{e}_i \in \F^n$, where $\textbf{e}_i$
is the $i$th standard-basis vector (i.e., the all-zeros vector
in $\F^n$ with a one in coordinate $i$).
The client would then split this vector into two additive shares
$w_A$ and $w_B$ such that $w_A + w_B = m \cdot \textbf{e}_i$, 
and send one of each of these ``write-request'' vectors to each
of the two servers. The servers then process the write by setting:
\[ D_A \gets D_A + w_A \in \F^n \qquad D_B \gets D_B + w_B \in \F^n,\]
which has the effect of adding the value $m \in \F$ into
the contents of the $i$th mailbox in the system.

The communication cost of this na\"ive approach is large: updating a single
mailbox requires the client to send $n$ field elements to each server. 

\noindent\textbf{Improving efficiency via DPFs}. Instead of sending such a large message, the client uses \emph{distributed 
point functions} (DPFs)~\cite{dpfs,dpfs2,fss} to compress these vectors. DPFs allow a client to split a point function $f$, in this case a function mapping \rev{indices} in the client's vector to their respective values, into two \emph{function shares} $f_A$ and $f_B$ which individually reveal nothing about $f$, but whose sum at any point is the corresponding value of $f$. More formally, let $f_{i^*,m}\colon [N] \to \F$ be a point function that evaluates to 0 at every point $i\in[N]$ except that $f(i^*) = m \in \F$. A DPF allows a client holding $f_{i^*,m}$ to generate shares $f_A$ and $f_B:[N]\to \F$ such that:
\begin{itemize}
  \item[(i)] an attacker who sees only one of the two shares learns nothing
            about $i^*$ or $m$, and
  \item[(ii)] for all $i \in [N]$, $f_{i^*,m}(i) = f_A(i) + f_B(i) \in \F$.
\end{itemize}
Moreover, in addition to supporting messages $m\in \F$, the latest generation of DPFs~\cite{dpfs2} 
allow for any message $m\in\{0,1\}^*$. 
When using these DPFs with
security parameter $\lambda$, each function share ($f_A$ and $f_B$) has 
bitlength $O(\lambda\log N + |m|)$. In addition to general improvements in efficiency over prior DPFs, our choice of DPF scheme will enable new techniques that we introduce in Section~\ref{sec:malicious}. 

In essence, the client can use DPFs to \textit{compress} the vectors
$w_A$ and $w_B$, which reduces the communication cost to 
$O(\lambda \cdot \log N + \log |\F|)$ bits, 
when instantiated with a pseudorandom function~\cite{GGM84}
using $\lambda$-bit keys. Upon receiving $f_A$ and $f_B$ the servers can evaluate them at each point $i\in [n]$ to recover the vectors $w_A$ and $w_B$ and update $D_A$ and $D_B$ as before. 

\subsection{Hiding Metadata without Synchronized Rounds}\label{accesscontrol}

Private writing alone does not suffice to provide metadata-hiding privacy. In order to achieve this, we also need to control \emph{read} access to mailboxes. Otherwise, a network adversary who controls a single client could read the contents of all mailboxes between each pair of writes and learn which client's message modified which mailbox contents, even if messages are encrypted. Prior works such as Pung~\cite{pung} or Riposte~\cite{riposte} prevent this attack by operating in batched rounds in which many clients write messages before any client is allowed to read. 
The key feature that allows \name to hide metadata without relying on synchronized rounds is that a message can only be read by the mailbox owner to whom it is sent. \name can make messages available to mailbox owners immediately as long as (1) the messages remain inaccessible to an attacker who does not own the mailbox whose contents have been modified and (2) the attacker cannot tell which mailbox has been modified if it does not own the modified mailbox. Thus, all we need to successfully hide metadata without rounds is a mechanism for access control that satisfies these two requirements. \rev{While an adversary who continuously reads from all mailboxes could then still learn \emph{when} a write occurs, it would learn nothing about which mailbox contents were modified as a result. }

\name includes a lightweight cryptographic approach to access control that relies on symmetric encryption, does not require the servers to undertake any user authentication logic when serving read requests, and enables useful implementation optimizations. 
A client registering a mailbox uploads keys $k_A$ and $k_B$ to servers $A$ and $B$ respectively, and the servers encrypt stored data using the respective key for each mailbox, decrypting before making modifications and re-encrypting after. The re-encryption ensures that the contents of every mailbox are rerandomized after each write, so an attacker attempting to read the contents of a mailbox for which it does not have both keys learns nothing from reading the encrypted contents of the mailbox, including whether or not those contents have changed. This property still holds even if only one of the two servers carries out the re-encryption, so its security is unaffected if a malicious server does not encrypt or re-encrypt mailboxes. Our implementation encrypts mailbox contents in counter mode, so re-encryption simply involves subtracting the encryption of the previous count and adding in the new one. Since these operations are commutative, we can implement an optimization where re-encryption is not done on every write but only \rev{when a read occurs after one or more writes}. This makes our approach -- which requires only symmetric encryption -- more efficient than a straightforward one based on public key encryption, e.g., where the contents of each mailbox are encrypted under the owner's public key when a read is requested. 

\section{Protecting Against Malicious Clients}\label{sec:malicious}

The techniques in Section~\ref{sec:tech} suffice to provide privacy if all clients behave honestly, but they are vulnerable to disruption by a malicious client. In the scheme described thus far, one malicious client can corrupt the state of the two servers with a single message. To do so, the malicious client sends DPF shares $f_A$ and $f_B$ to the servers that expand into vectors $w_A$ and $w_B$ such that $w_A + w_B = v \in \F^n$, where $v$ is non-zero at many (or even all) coordinates. 
A client who submits such DPF key shares can, with one message to the servers, write
into \emph{every} mailbox in the system, 
corrupting whatever actual messages each mailbox may have held. 

\name protects against this attack with an \emph{auditing} protocol that checks to make sure $(w_A + w_B) \in \F^n$ is a vector with at most one non-zero component. In other words, the servers check that each write request updates only a single mailbox. Any write request that fails this check can be discarded to prevent it form corrupting the contents of $D_A$ and $D_B$. 
Riposte~\cite{riposte}, a prior work that also audits DPFs to protect against malicious clients, uses a three-server auditing protocol that 
requires communication $\Omega(\lambda \sqrt n)$ and client computation $\Omega(\sqrt n)$
for a system with $n$ mailboxes, where $\lambda$ is a security parameter. However, their protocol takes advantage of the structure of a particular DPF construction that is less efficient than the one used by~\name. Applying their protocol to the more efficient DPFs used in \name would require client communication and computation $\Omega(\lambda n)$ and $\Omega(n)$ respectively as well as the introduction of an
additional non-colluding server.
This linear bandwidth consumption \textit{per write} would create a communication bottleneck in \name and increase client-side computation costs significantly. Moreover, adding a third server -- and requiring that two out of three servers remain honest to guarantee security -- would dramatically reduce the practicality of the \name system. To resolve this issue, we introduce a new auditing protocol that drops client computation (in terms of AES evaluations and finite field operations) to $O(1)$ and communication to $O(\lambda)$ while simultaneously \emph{eliminating the need for a third server} to perform audits. We describe our two-party auditing protocol in Section~\ref{sec:tech:audit}. 

Although auditing ensures that DPFs sent by clients must be well-formed, an attacker targeting \name has a second avenue to disrupting the system. 
Instead of attempting to corrupt the entire set of mailboxes -- an attack prevented by the auditing protocol -- a malicious client can write random data to \emph{only one mailbox} and corrupt any message a source may send to a journalist over that mailbox. Although this attack is easily detectable when a journalist receives a random message, it still allows for easy disruption of the system and cannot be blocked by blind auditing because the disruptive message is structured as a legitimate write. 

We defend against this kind of targeted disruption with a new application of \emph{virtual addressing}. At a high level, we assign each mailbox a unique 128-bit virtual address and modify the system to ensure that writing into a mailbox requires knowing the mailbox's virtual address.
In this way, a malicious user cannot corrupt the contents of an honest user's
mailbox, since the malicious user will not be able to guess the honest 
user's virtual address. We discuss this defense and its implications for other components of the system in Section~\ref{virtaddrs}.

\subsection{Auditing to Prevent Disruption}\label{sec:tech:audit}

This section describes our auditing protocol. We begin with a rough outline of the protocol before stating the security properties required of it and then explaining the protocol in full detail. At a high level, our auditing protocol combines the verifiable DPF protocol of Boyle et al.~\cite{dpfs2}, which only provides security against \emph{semi-honest} servers, with \emph{secret-shared non-interactive proofs} (SNIPs) first introduced by the Prio system~\cite{prio} (and later improved and generalized by Boneh et al.~\cite{BBC+19}) to achieve security against \emph{fully malicious} servers. We explain each of these ideas and how we combine them below. 

Let the vectors $w_A$ and $w_B\in\F^n$ be the outputs that servers $A$ and $B$ recover after evaluating $f_A(i)$, $f_B(i)$,  for $i\in[n]$. Note that even DPFs that output a message in $\{0,1\}^*$ begin with an element of a $\lambda$-bit field $\F$ and expand it, so for the purposes of our auditing protocol, we can assume that every DPF output is an element of $\F$. 
We say that $w = w_A + w_B \in \F^n$ is a \emph{valid} write-request vector if it is a vector in $\F^n$ of Hamming-weight at most one. The goal of the auditing protocol is to determine whether a given write-request vector is valid. 

The observation of Boyle et al.~\cite{dpfs2} is that the following $n$-variate 
polynomial equals zero with high probability over the random choices of $r_1,...,r_n$ if and only if (1) there is at most one nonzero $w_i$ and (2) $m=w_i$ for the nonzero value of~$w_i$
\begin{align*}
f(r_1,...,r_n)=(\Sigma_{i\in[n]}w_ir_i)^2-m\cdot(\Sigma_{i\in [n]}w_ir_i^2). 
\end{align*}

This polynomial roughly corresponds to taking a random linear combination of the elements of $w$ -- using randomness shared between the two servers -- and checking that the square of the linear combination and the sum of the terms of the linear combination squared are the same. Using the fact that it is easy to compute linear functions on secret-shared data, the two sums in the equation above can be computed non-interactively by servers $A$ and $B$. Boyle et al. suggest using a multiparty computation between the servers to compute the remaining multiplications and check whether this polynomial in fact equals zero, thus determining whether the DPF is valid.

The problem with this approach is that it is only secure against \emph{semi-honest} servers. A malicious server can deviate from the protocol and potentially learn which entry of $w$ is non-zero. For example, suppose a malicious server $A$ is interested in knowing whether a write request modifies an index $i^*$. It runs the auditing protocol as described, but it replaces its value $w_{Ai^*}$ with a randomly chosen value $w_{Ai^*}'$. If $w_{Ai^*}+w_{Bi^*}=0$, i.e., $i^*$ was not the nonzero index of $w$, this modification will cause the audit to fail because the vector $w'$ that includes $w_{Ai^*}'$ instead of $w_{Ai^*}$ no longer has hamming weight one. Thus the malicious server learns that the write request would not have modified index $i^*$. On the other hand, if $w_{Ai^*}+w_{Bi^*}\neq0$, i.e., $i^*$ was the nonzero index of $w$, the inclusion of $w_{Ai^*}'$ still results in a vector $w'$ of hamming weight one, and the auditing protocol passes. Thus the malicious server can detect whether or not the write request modifies index $i^*$ by observing whether or not auditing was successful after it tampers with its inputs. 

To prevent this attack we make use of a SNIP proof system~\cite{prio,BBC+19}. In a SNIP, a client sends each server a share of an input $w$ and an arithmetic circuit $\textsf{Verify}()$. The client then uses a SNIP proof to convince the servers, who only hold shares of $w$ but may communicate with each other, that $\textsf{Verify}(w)=1$. An important property of a SNIP proof system is that it provides security against \emph{malicious} servers. That is, even a server who deviates from the protocol cannot abuse a SNIP to learn more about $w$. SNIP proofs require computation and communication linear in the number of multiplications between secret values in the statement being proved. Our approach is to instantiate the DPF verification protocol of Boyle et al.~\cite{dpfs2} \emph{inside} of a SNIP to protect it from potentially malicious servers. Since the Boyle et al. verification protocol only requires two multiplications between shared values, the squaring and the multiplication by $m$, this results in a constant-sized SNIP (i.e. size $O(\lambda)$). 

\noindent\textbf{Properties of auditing protocol.}
Before describing our protocol in detail, we recall the completeness,
soundness, and zero-knowledge properties we require of the auditing
protocol (adapted from those of 
Riposte's auditing protocol~\cite{riposte}).

\begin{denseitemize}
\item \textbf{Completeness.}
  If all parties are honest, the audit always accepts. 

\item \textbf{Soundness against malicious clients.}
  If $w$ is not a valid write request (i.e., the client is malicious) 
  and both servers are honest, 
  then the audit will reject with overwhelming probability.

\item \textbf{Zero knowledge against malicious server.}
  Informally: as long as the client is honest, 
  an active attacker controlling at most one server 
  learns nothing about the write request $w$, apart from the fact that it is valid.
  That is, for any malicious server there exists an efficient algorithm
  that simulates the view of the protocol execution with
  an honest second server and an honest client.
  The simulator takes as input only the public system parameters
  and the identity of the malicious server.
\end{denseitemize}

\noindent\textbf{Our auditing protocol.} 
Our auditing protocol proceeds as follows. 
We assume that data servers $A$ and $B$ share a private stream of
random bits generated from a pseudorandom generator with a seed $r$.
In practice, the servers 
generate the random seed by agreeing on a shared secret at setup and using
a pseudorandom generator to get a new seed for each execution of this protocol. We will describe the protocol using a SNIP as a black box and give details on how to instantiate the SNIP in Appendix~\ref{app:snips}.

At the start of the protocol, server $A$ holds $r$ and $w_A \in \F^n$ and 
server $B$ holds $r$ and $w_B \in \F^n$, both generated by evaluating the DPF shares sent by the client at each registered mailbox address. The client holds the index $i^*$ at which $w$ is non-zero as well as the values of $w_A$ and $w_B$ at index $i^*$, which it computes from the function shares $f_A$ and $f_B$ that it sent to the servers. 
\begin{enumerate}

\item \emph{Servers derive proof inputs.}\\ The servers begin by sending the random seed $r$ used to generate their shared randomness to the client. 

Next, they compute shares $m_A$ and $m_B$ of $m$, the value of $w$ at its non-zero entry, which is simply the sum of all the elements of $w_A$ or $w_B$ respectively because all but one entry of $w$ should be zero. That is, the servers compute
\begin{align*}
m_A\gets\Sigma_{i\in[n]}w_{Ai}\quad\text{and}\quad m_B\gets\Sigma_{i\in[n]}w_{Bi}.
\end{align*}
Then servers $A$ and $B$ use their shared randomness $r$ to generate a random vector $r=(r_1,...,r_n)\in\F^n$ and then compute the vector of squares $R=(r_1^2,...,r_n^2)\in\F^n$. After this, they compute shares of the ``check'' values $c=\langle w,r\rangle$ and $C=\langle w,R\rangle$:
      \begin{align*}
        c_A \gets \langle w_A,r \rangle \in \F, \quad 
        C_A \gets \langle w_A,R \rangle \in \F\\
        c_B \gets \langle w_B,r \rangle \in \F, \quad
        C_B \gets \langle w_B,R \rangle \in \F
      \end{align*}
      Here the notation $\langle x, y\rangle$ represents the inner product between vectors $x,y\in\mathbb{F}^n$, defined as $\Sigma_{i=1}^{n}x_iy_i$.
      
      At this point, the servers hold values $m_A, c_A, C_A$ and $m_B,c_B,C_B$ respectively. 
      
    \item \emph{Client derives proof inputs.}\\Since the client knows the seed $r$, the index $i^*$, and the values of $w_A$ and $w_B$ at index $i^*$ (and as a consequence the value of $m=w_{Ai^*}+w_{Bi^*}\in\F$), the client can compute the random values $r^*,r^{*2}$ that will be multiplied by the $i^*$th entries of $w_A$ and $w_B$. Since all the values other than the $i^*$th entry of $w$ are zero, the client need not compute them.
  Thus the client computes the check values $c^*=r^*\cdot (w_{Ai^*}+\cdot w_{Bi^*})$ and $C^*=r^{*2}\cdot (w_{Ai^*}+\cdot w_{Bi^*})$. Note that this allows the client to compute the check values in only $O(1)$ time even though the servers must do $O(n)$ work to find them. 

\item \emph{Proof computation and verification.}\\To complete the proof, the client prepares a SNIP proof $\pi=(\pi_A, \pi_B)$, sends $\pi_A$ to server $A$, and sends $\pi_B$ to server $B$. The servers then verify the proof, communicating with each other as needed. The SNIP proves that 
\begin{align*}
c^{2}-m\cdot C=0
\end{align*}
where $c\gets c_A+c_B$ and $C\gets C_A+C_B$. 

The soundness property of the SNIP proof guarantees that the servers will only accept the proof if the statement is true, and the zero-knowledge property of the proof guarantees that as long as one server is honest, the servers learn nothing from receiving the SNIP proof that they did not know before receiving it (even if one server is fully malicious). Note that this statement only involves two multiplications: $c\cdot c$ and $m\cdot C$.  
\end{enumerate}

We sketch the instantiation of the proofs used in our auditing protocol as well as the security analysis of the full auditing protocol in Appendix~\ref{app:snips}. Full details and a security proof for the SNIP proof system itself can be found in the Prio paper~\cite{prio} and the follow-up work of Boneh et al.~\cite{BBC+19}.

\subsection{Preventing Targeted Disruption}\label{virtaddrs}

We now describe how \name prevents a targeted attack where a malicious client writes random data to a single mailbox to corrupt its contents. \name servers assign each mailbox a 128-bit \emph{virtual address} and ensure that a client can only write to a mailbox if it knows the corresponding virtual address. 

To implement this, the \name servers maintain an array of $n$ \textit{physical} mailboxes,
but they also maintain an array of $2^\lambda$ \textit{virtual} mailboxes,
where $\lambda \approx 128$ is a security parameter.
The two data servers assign a unique virtual address to each physical mailbox, and
they collectively maintain a mapping -- a \textit{page table} -- 
that maps each active virtual address to a physical mailbox.
Since the virtual addressing scheme's only goal is to prevent misbehavior by malicious clients, the servers
both hold the contents of the page table (i.e., the list of active virtual
addresses and their mapping to physical addresses) in the clear.
The virtual-address space (around $2^{128}$ entries) 
is vastly larger than the number of physical mailboxes (around $2^{20}$, perhaps),
so the vast majority of the virtual-address space goes unused.

When a client registers a new mailbox, the servers
both allocate storage for a new physical mailbox, assign a new random
virtual address to this physical mailbox, and update their page tables. 
The address can either be chosen by one server and sent to the other or generated separately by each server using shared randomness. 
The servers then return the virtual and physical addresses for the mailbox to the client. 
As mentioned above, a mailbox owner must communicate its address to others to receive messages. We describe how this can be achieved when we discuss dialing in Section~\ref{dialing}. The contents of the tables stored at the servers are shown in Figure~\ref{fig:virtaddr}.

When preparing a write request, the client prepares DPF shares $f_A$ and $f_B:2^\lambda\to \F$ as if it were going to write in to the exponentially large address space. However, instead of evaluating shares at every $i\in[2^\lambda]$, the \name servers only evaluate $f_A$ and $f_B$ at the currently active virtual addresses. In this way, the number of DPF evaluations the servers compute remains linear in the number of registered mailboxes, even though clients send write requests as if the address space were exponentially large. A client who does not know the address for a given mailbox has a chance negligible in $\lambda$ of guessing the correct virtual address. Note that this technique is only possible because \name uses a DPF whose share sizes are \emph{logarithmic} in the function domain size. Using virtual addresses with older square-root DPFs would result in infeasibly large message sizes and computation costs.

Although virtual addressing, when combined with auditing, does fully resolve the issue of disruptive writes, it does not fully abstract away physical addresses. Our auditing protocol critically relies on the client knowing the index of the mailbox it wants to write to among the set of all mailboxes. As such, a client preparing to send a message must be informed of both the virtual and physical addresses of the mailbox it wishes to write to. Fortunately, the size of a physical address is much smaller than that of a virtual address (about 20 bits  compared to 128 bits for a virtual address), so communicating both addresses at once adds little cost to only sending the virtual address.

\section{Full \name Protocol}\label{sec:full}

\begin{figure}
\centering 
\includegraphics[width=.8\linewidth]{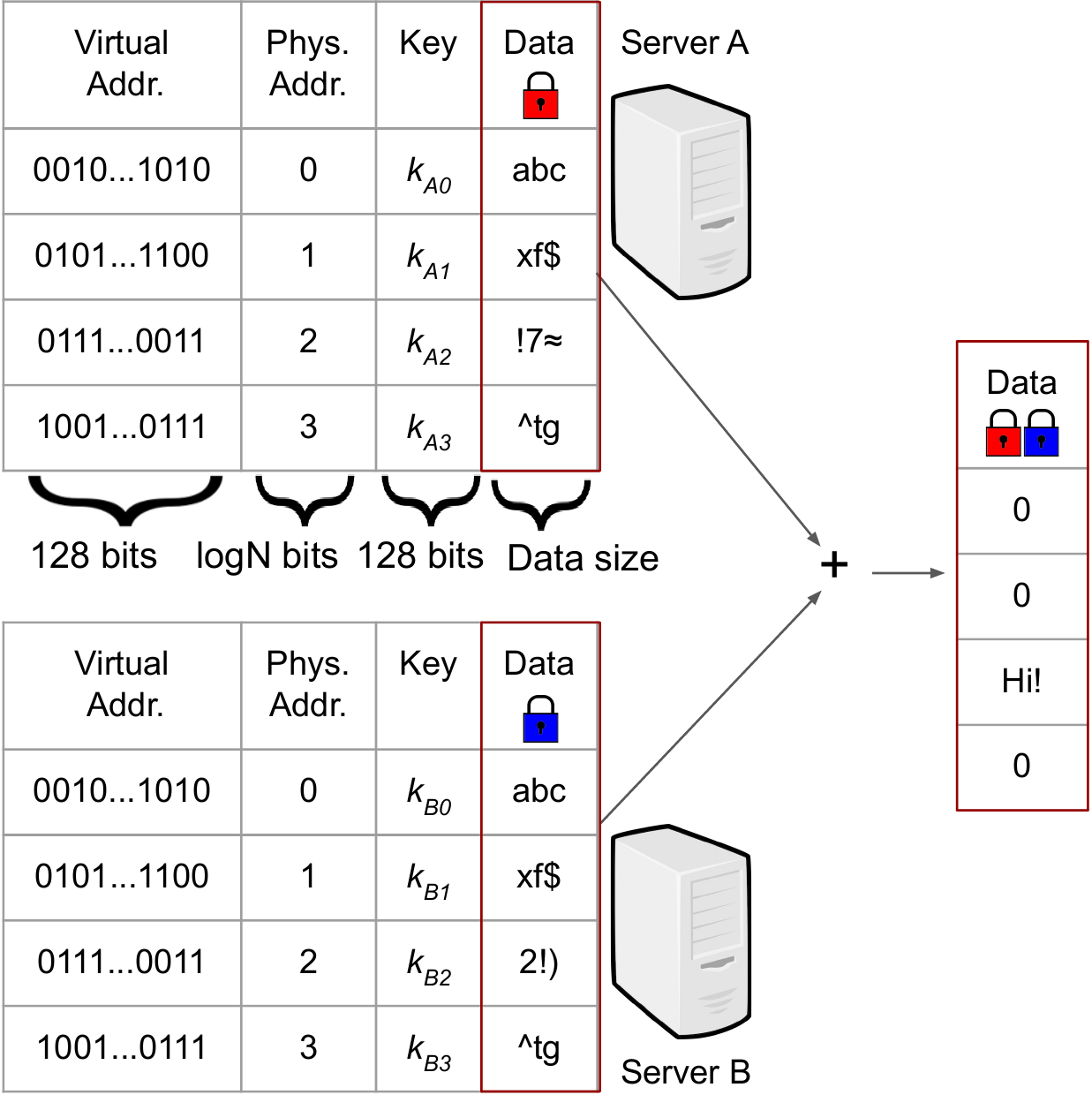}
\caption{\small Contents of the tables held by servers in \name. Each server stores the conversion from virtual to physical addresses and a distinct key for each mailbox. Combining data from the two servers allows a user holding both keys for a given mailbox to read its contents. }
\label{fig:virtaddr}
\end{figure}

This section summarizes the full \name protocol described incrementally in Sections~\ref{sec:tech} and~\ref{sec:malicious}. We will describe the protocol in full but refer to the steps of the auditing protocol as described in Section~\ref{sec:tech:audit} to avoid repeating the protocol spelled out in detail there. We prove security in Appendix~\ref{app:security}. After describing the protocol, we describe how clients can add message integrity to their \name messages. 

We assume that a mailbox owner has already set up a mailbox with virtual address $v$ and physical address $p$ and communicated $(p,v)$ to another client. We discuss options for communicating $p$ and $v$ to other clients (``dialing'') in Section~\ref{dialing}. We also assume that the mailbox owner holds mailbox keys $k_A$ and $k_B$, which it has sent to servers $A$ and $B$ respectively, and the client has a message $m$ that it wants to send. 
Server $A$ holds vectors $V$ of virtual addresses, $K_A$ of keys, and $D_A$ of mailbox contents, each of length $n$. Server $B$ likewise holds $V$, $K_B$ and $D_B$. Each entry of $D_A$ and $D_B$ is encrypted in counter mode under the corresponding key in $K_A$ or $K_B$. 
Figure~\ref{fig:virtaddr} shows the information held by servers $A$ and $B$ for each mailbox. 

\noindent\textbf{Sending a message}. 
\begin{enumerate}
\item The client generates DPF shares $f_A$ and $f_B$ of the point function $f_{v,m}:[2^\lambda]\to \{0,1\}^{|m|}$. It sends $f_A$ to $A$ and $f_B$ to~$B$. 

\item $A$ and $B$ evaluate $w_A\gets (f_A(V_1), ..., f_A(V_n))$ and $w_B\gets (f_B(V_1), ..., f_B(V_n))$. They use their shared randomness to generate a seed $r$ to be used in the auditing protocol, send it to the client, and prepare the server inputs to the SNIP. 

\item The client prepares the client inputs to the SNIP and generates the corresponding proof $\pi=(\pi_A,\pi_B)$. It sends $\pi_A$ to server $A$ and $\pi_B$ to server $B$. 

\item The servers verify the SNIP proof $\pi$, and they abort if the verification fails. 

\item Servers $A$ and $B$ decrypt each $D_{Ai}$ with $K_{Ai}$ and each $D_{Bi}$ with key $K_{Bi}, i\in[n]$. Next, they set $D_{Ai}\gets D_{Ai}+w_{Ai}$ and $D_{Bi}\gets D_{Bi}+w_{Bi}$ before re-encrypting the new values of $D_{Ai}$ and $D_{Bi}$ under the same keys (with new nonces).

\end{enumerate}

\noindent\textbf{Checking a mailbox}. 
\begin{enumerate}
\item The mailbox owner sends $(p,v)$ to servers $A$ and $B$ to request to read from the mailbox at physical address $p$.

\item Servers $A$ and $B$ check that virtual address $v$ corresponds to physical address $p$ and then send $D_{Ap}$ and $D_{Bp}$ as well as the nonce used for the encryption of each value. Then they set the values of $D_{Ap}$ and $D_{Bp}$ to fresh encryptions of $0$ under $K_{Ap}$ and $K_{Bp}$ respectively, emptying the mailbox. Since only the mailbox owner and whoever wrote into a mailbox know $p$ and $v$, and the virtual address space for $v$ is huge, clients cannot read or delete the contents of each other's mailboxes. 

\item The mailbox owner decrypts the values of $D_{Ap}$ and $D_{Bp}$ it received with keys $k_A$ and $k_B$ to get messages $m_{Ap}$ and $m_{Bp}$. It outputs message $m \gets m_{Ap} + m_{Bp}$. 
\end{enumerate}

\noindent\textbf{Complexity}. \rev{Table}~\ref{complexityFig} shows the communication and computational complexity of sending a message in \name for the client and the servers. We measure computational complexity in terms of AES evaluations and field operations separately to better capture the computation being carried out by each party. The complexities reported are the sum of costs due to DPF evaluation, re-encryption, and auditing. 

Client communication includes sending a DPF whose shares are functions with domain size $2^\lambda$, resulting in DPFs of size $O(\lambda^2+|m|)$. As discussed in Section~\ref{sec:tech:audit}, the auditing protocol involves the client sending a proof of size $O(\lambda)$.

Cryptographic costs on the client include generating DPF shares and evaluating the DPF at one point, both of which cost $O(\lambda+|m|)$. The server, on the other hand, must evaluate the DPF at each address and also generate the random vectors needed for the auditing protocol. The number of field operations for each party come directly from the costs incurred during the auditing protocol. 

\begin{table}\small
\centering
\begin{tabular}{rll}\toprule
&Client&Servers\\\midrule
Communication&$O(\lambda^2+|m|)$&$O(\lambda)$\\
AES Evaluations&$O(\lambda + |m|)$&$O(n(\lambda+|m|))$\\
Field Operations&$O(1)$&$O(n)$\\\bottomrule
\end{tabular}
  \caption{\small Complexity of processing a single write in \name with $n$ mailboxes, message size |m|, and security parameter $\lambda$. Communication measures bits sent only. }
\label{complexityFig}
\end{table}

\paragraph{\rev{Message integrity}}. \rev{The core \name protocol does not protect message integrity, so a malicious server could undetectably corrupt the contents of a mailbox. This can be remedied in a straightforward way by using MACs. Given that the clients writing to and reading from a mailbox share a secret to establish an address, they could instead use a master secret to derive (e.g., via a hash) a mailbox address and a MAC key. Messages written to \name could then be MACed before being split into shares via a DPF. Since a MAC-then-encrypt approach provides authenticated encryption when the encryption is done in counter mode~\cite{cryptobook} (as we do), \name with MACed messages provides authenticated encryption on the messages.}

\section{Using \name for Whistleblowing}\label{sec:dd}

Having described the core \name system itself, this section covers two important considerations involved in using \name for whistleblowing: plausible deniability for whistleblowers and agreeing on mailbox addresses.

First, in order to provide meaningful security in practice, \name must hide both the recipient of a given client's message as well as \emph{whether} a client is really communicating with a
journalist. We discuss how to provide plausible deniability for \name clients in Section~\ref{covertraffic}. Second, to set up their communication channel, a journalist and whistleblower
must agree on a mailbox address through which they will communicate. This
can be done either in person or via a dialing protocol as described in Section~\ref{dialing}. 

\subsection{Plausible Deniability}\label{covertraffic} 
We now turn to
the goal of hiding whether or not a client is really communicating with a
journalist. If \name were only to be used by journalists and their sources, it
would fundamentally fail to serve its purpose. Although no observer could
determine which journalist a given message was sent to, the mere fact that
someone sent a message using \name reveals that she must be a source for some
journalist. In order to provide plausible deniability to whistleblowers, 
other, non-whistleblowing users must send messages through the system as 
well. 
 
One solution for this problem, first suggested in the Conscript system~\cite{conscript}, is to have cooperative web sites embed Javascript in their pages that generates and submits dummy requests. For example, the New York Times home page could be modified such that each time a consenting user visits (or for every $n$th consenting user that visits), Javascript in the page directs the browser to generate a request to a special write-only \name dummy address that the servers maintain but for which each server generates its own encryption key not known to any user. Since no user has the keys to unlock this address, messages written to it can never be retrieved, and \name's metadata-hiding property guarantees that messages sent to the dummy address are indistinguishable from real messages sent to journalists. This enables creating a great deal of cover traffic and gives clients who really are whistleblowers plausible deniability, \rev{as long as communication patterns between users and the \name servers are the same for real and cover traffic}. Moreover, only one large organization needs to implement this technique for all news organizations who receive messages through \name to benefit from the cover traffic. 
The exact quantity of cover traffic required to provide the appropriate level of protection for whistleblowers using \name is ultimately a subjective decision, but the \name metadata-hiding guarantee implies that a whistleblower sending a message through \name cannot be distinguished among the set of all users sending messages through \name, be they real messages or cover messages. 

\name is particularly well-suited to this approach for two reasons: aligned incentives and low client side costs. First, participating news organizations all have web sites and a natural incentive to direct cover traffic to the \name system. \rev{Even if only one or a few organizations among them are willing to risk adding dummy traffic scripts to their pages, everyone benefits. In fact, even the same organizations who are willing to host the \name servers could add the dummy scripts to their own news websites to ensure adequate cover traffic.} Second, as demonstrated in Section~\ref{eval}, \name's extremely low client computation and communication requirements lend themselves particularly well to this approach, since the client can easily run in the background on a web browser, even in computation or data-restricted settings such as mobile devices. We empirically evaluate a JavaScript version of the \name client in Section~\ref{clienteval} and find it imposes very little additional cost on the browser.


Using in-browser JavaScript to give users plausible deniability raises a number of security and ethical concerns. We defer to the Conscript paper~\cite{conscript} for an extensive discussion of the security and ethical considerations involved and note that it is also possible to generate cover traffic for \name using a standalone client, as is common in other systems. 

\subsection{Dialing}\label{dialing}

In order to use \name, a journalist and source must agree on the mailbox address which the source will use to send messages to the journalist. Journalists who make initial in-person contact with sources could, for example, distribute
business cards with mailbox addresses on them in QR code form. 

Journalists and sources could also use a more expensive \emph{dialing} protocol to share an initial secret before moving to \name to more efficiently communicate longer or more frequent messages. 
\rev{One approach to dialing that can conveniently integrate with \name is to use an improved version of the Riposte~\cite{riposte} system as a dialing protocol. Riposte offers a public broadcast functionality that progresses in fixed time epochs, where anyone can announce a message to the world. Since journalists can easily post their public keys online, e.g., next to their name at the bottom of articles they write, anyone wishing to connect with a particular journalist can send a mailbox address (and perhaps some introductory text) encrypted under that journalist's public key with no other identifying information. A client run by a journalist can download all Riposte messages sent in a day and identify those encrypted under that journalist's public key. The journalist can then register any mailbox addresses sent to it and communicate with whoever sent the messages via \name.} This requires mailbox owners \rev{(in this case, the journalist)} to choose virtual addresses instead of the servers, but the probability of colliding addresses is low because the virtual address space is large. \rev{Using this approach to dialing gives \name users the ability to bootstrap from a single message in a dialing system with fixed-duration rounds to as many messages as they want in a system which processes messages asynchronously.}

\rev{Since Riposte has a similar underlying architecture to \name, a number of the techniques used in \name could be used to make it a more effective dialing protocol. Most importantly, instead of using Riposte's DPFs and auditing protocol, which are less efficient and require a third non-colluding server, the dialing protocol can use a Riposte/\name hybrid approach where the DPF and auditing protocol are those of \name. This means that the dialing protocol relies on the same trust assumptions as the main protocol, and it can even be deployed on the same servers.} 

\rev{Integrity in the dialing protocol can be ensured in a way similar to the main protocol as well. Instead of sending only a mailbox address, clients send a secret from which a mailbox address and MAC key can be derived, and the encrypted message is then MACed using that key. To ensure that servers can't tamper with or erase messages by changing their state after seeing that of the other server, they are required to publish and send each other commitments to (hashes of) the message shares they hold before publishing the actual databases of messages.}

\section{Implementation and Evaluation}\label{eval}

\begin{figure*}
\centering
\minipage{0.32\textwidth}
\begin{tikzpicture}
\begin{loglogaxis}[
    title style={align=center}, title={Server Communication},
    width=6cm, height=4.5cm,
    xlabel={Number of Mailboxes},
	y label style={at={(axis description cs:.1,.5)},anchor=south},
	ylabel={Communication [KB]},
    cycle list name=auto,
    ymin=.1, ymax=90000,
    xtick={100,1000,10000,100000,1000000},
     tick pos=left,
     minor tick style={draw=none},
    ymajorgrids=true,
    grid style=dashed,
    	legend style={at={(1.05,-.35)}},
		legend entries={\small Pung, \small Riposte, \small \name},
	legend columns=3,
    ]
    \addplot[color=blue, mark=*]
    coordinates {
    (64,215.3)(256,319.77)(1024,541.0)(4096,868.6)(16384,1671.4)(65536,3850.5)(262144,11354.4)(1048576,38617.4)
    };
    \addplot[color=red, mark=square*]
    coordinates {
    (64,18.397)(256,31.358)(1024,56.433)(4096,106.503)(16384,208.241)(65536,409.437)(262144,815.596)(1048576,1624.3)
    };
    \addplot[color=brown, mark=diamond*]
    coordinates {
    (64,8.338)(256,8.338)(1024,8.338)(4096,8.338)(16384,8.338)(65536,8.338)(262144,8.338)(1048576,8.338)
    };
\end{loglogaxis}
\end{tikzpicture}
\caption{Server communication costs when sending 160 Byte messages, including both data sent and received. Riposte also requires an auditing server whose costs are not depicted.}
\label{servercommunicationgraph}

\endminipage\hfill
\minipage{0.32\textwidth}
\begin{tikzpicture}
\begin{loglogaxis}[
    title style={align=center}, title={\\Client Communication},
    width=6cm, height=4.5cm,
    xlabel={Number of Mailboxes},
	y label style={at={(axis description cs:.1,.5)},anchor=south},
	ylabel={Communication [KB]},
    cycle list name=auto,
    ymin=.1, ymax=90000,
    xtick={100,1000,10000,100000,1000000},
    tick pos=left,
    minor tick style={draw=none},
    legend pos=south east,
    ymajorgrids=true,
    grid style=dashed,
    	legend style={at={(1.05,-.55)}},
		legend entries={\small Pung, \small Riposte, \small \name},
	legend columns=3,
    ]
    \addplot[color=blue, mark=*]
    coordinates {
    (64,215.3)(256,319.77)(1024,541.0)(4096,868.6)(16384,1671.4)(65536,3850.5)(262144,11354.4)(1048576,38617.4)
    };
    \addplot[color=red, mark=square*]
    coordinates {
    (64,5.246)(256,9.614)(1024,18.045)(4096,34.871)(16384,69.080)(65536,137.017)(262144,273.778)(1048576,545.881)
    };
    \addplot[color=brown, mark=diamond*]
    coordinates {
    (64,5.393)(256,5.393)(1024,5.393)(4096,5.393)(16384,5.393)(65536,5.393)(262144,5.393)(1048576,5.393)
    };
\end{loglogaxis}
\end{tikzpicture}
\caption{Client communication costs when sending 160 Byte messages, including both data sent and received. \name requires significantly less communication than prior work. }
\label{clientcommunicationgraph}
\endminipage\hfill
\minipage{0.32\textwidth}
\begin{tikzpicture}
\begin{loglogaxis}[
    title style={align=center}, title={Audit Computation},
    width=\linewidth, height=3.5cm,
    xlabel={Number of Mailboxes},
	y label style={at={(axis description cs:.1,.5)},anchor=south},
	ylabel={Time [ms]},
    cycle list name=auto,
    ymin=.0001, 
    xtick={100,1000,10000,100000,1000000},
    tick pos=left,
    minor tick style={draw=none},
    legend pos=south east,
    legend style={at={(1.05,-1.25)}},
	legend entries={\small Riposte Server, \small \name Server, \small Riposte Client, \small \name Client, \small Riposte Auditor},
	legend columns=2,
    ymajorgrids=true,
    grid style=dashed,
    ]
    \addplot[color=red, mark=square*]
    coordinates {
    (1000,.0554)(10000,.4958)(100000,5.0234)(1000000,51.5278)
    };
    \addplot[color=brown, mark=diamond*]
    coordinates {
    (1000,.1628)(10000,1.3734)(100000,13.5726)(1000000,135.626)
    };
    \addplot[color=red, dashed, mark=square]
    coordinates {
    (1000,.193)(10000,1.569)(100000,15.4236)(1000000,153.9422)
    };
    \addplot[color=brown, dashed, mark=diamond*]
    coordinates {
    (1000,.005)(10000,.002)(100000,.003)(1000000,.003)
    };
    \addplot[color=red, dotted, mark=square] 
    coordinates {
    (1000,.1032)(10000,.9592)(100000,9.6256)(1000000,96.3496)
    };
\end{loglogaxis}
\end{tikzpicture}

\caption{Our auditing protocol dramatically reduces computation costs for the client while server-side costs remain comparable to prior work, where audit computation time is dwarfed by DPF evaluation anyway.}
\label{auditgraph}
\endminipage
\end{figure*}

We implement \name with the underlying cryptographic operations (DPFs, auditing) in C and the higher level functionality (servers, client) in Go. We use OpenSSL for cryptographic operations in C and base our DPF implementation in part on libdpf~\cite{libdpf}, which is in turn based on libfss~\cite{splinter,libfss}. We also re-implemented the client-side computations involved in sending a write request in JavaScript for the whistleblowing application, using the SJCL~\cite{sjcl,sjclpaper} and TweetNaCl.js~\cite{nacljs} libraries for crypto operations. 
We implement the DPF construction~\cite{dpfs2} and the auditing protocol using the field $\F_p$ of integers modulo the prime $p = 2^{128}-159$, since
these field elements have a convenient representation in two 64-bit words. 
Our implementation does not include the client-side integrity checks described in Section~\ref{sec:full}, but these checks can be added by clients with no impact on server-side code or performance. 

We evaluate \name on three Google Cloud instances (two running the servers and a third to simulate clients) with 16-core intel Xeon processors (Haswell or later) with 64GB of RAM each and 15.6~Gbps bandwidth. We run all three in the same datacenter to minimize network latency and focus comparisons to other systems on computational costs since we begin our evaluation by considering communication separately. We evaluate the JavaScript implementation of the whistleblowing client on a laptop with an Intel i5-2540M CPU @ 2.60GHz and 4GB of RAM running Arch Linux and the Chromium web browser. 
All experiments use security parameter $\lambda=128$. 

We compare \name to Riposte~\cite{riposte} and Pung~\cite{pung}, two prior works that also provide cryptographic metadata-hiding guarantees, albeit in slightly different settings. We choose to compare to these systems because, like \name, they also provide cryptographic security guarantees and only rely on a small number of servers to provide their security guarantees. Riposte requires 3 servers, of which two must be honest (a stronger trust assumption than \name) whereas Pung requires only a single server which can potentially be malicious (a weaker trust assumption).  We rerun the original implementations of Riposte and Pung on the same cloud instances used to evaluate \name. \rev{Our evaluation results do not distinguish between real and dummy messages because the two are identical from a performance perspective.}

We find that \name reduces communication costs by orders of magnitude compared to Riposte and Pung, with clients using over $100\times$ less bandwidth than Riposte and over $4000\times$ less bandwidth than Pung when sending a message in the presence of one million registered mailboxes. On the client implemented in C/Go, \name requires 20ms of computation to send a write request, even in the presence of one million registered mailboxes, and our JavaScript client performs similarly, requiring 51ms for the same task. 

We compare the performance of our auditing protocol to the prior protocol proposed by Riposte~\cite{riposte}. Despite making a weaker trust assumption and requiring only two servers, our protocol reduces client computation time by several orders of magnitude, resulting in audit compute time of under 5~\emph{micro}seconds regardless of the number of registered mailboxes and reducing overall client compute costs by $8\times$ compared to an implementation that uses Riposte's auditing protocol.

On the server side, we show that \name's throughput and latency costs are better than prior work. We also calculate the dollar cost of running each system to send one million messages and find that \name costs $6\times$ less to operate than Riposte, the second cheapest system. Throughout our experiments we generally compare to prior work on message sizes comparable to or larger than those used in their original evaluations. Since the recent whistleblower's report to the US intelligence community's inspector general contained 25.3KB of text~\cite{whistleblowingdoc} and last year's widely reported anonymous op-ed in the New York Times contained about 9KB of text~\cite{oped}, we make sure to evaluate \name on 32KB messages as well. 

\subsection{Communication Costs} 
Figures~\ref{servercommunicationgraph} and~\ref{clientcommunicationgraph} show communication costs for each party when sending a 160 Byte message and compares to costs in Riposte~\cite{riposte} and Pung~\cite{pung}. We use a smaller message size than in our subsequent experiments to focus on measuring the role of the DPF and auditing in communication costs. Communication costs always increase linearly with the size of the messages being sent. \name's communication costs are \emph{constant} regardless of the number of mailboxes, compared to asymptotically $\sqrt{n}$ in Riposte, the system with the next lowest costs. For $2^{14}$ mailboxes, \name has 8.34KB of communication by the server and 5.39KB by the client for each write. The corresponding costs in Riposte are 208KB and 69KB, respectively, representing communication reductions of $25\times$ on the server side and $13\times$ on the client. Riposte additionally requires a third auditing server which incurs 13.8KB of communication, whereas \name has no such requirement.
For about one million ($2^{20}$) mailboxes, \name requires $101\times$ less communication than Riposte on the client side and $195\times$ less on the server side. The communication reduction compared to Pung in this setting is $4,631\times$ on the server side and $7,161\times$ on the client side, reflecting the high cost of providing security with only one server as Pung does. 
Our communication savings come from using log-sized DPFs that write into a large but fixed-size virtual address space for write requests and from our new auditing protocol whose communication costs do not increase with the number of mailboxes.

\subsection{Client Costs}\label{clienteval}

Client computation time in both our native C/Go and in-browser Javascript implementations remains constant as the number of mailboxes on the server side increases: since the client always prepares a DPF to be run on the $2^{128}$-sized virtual address space, the cost of preparing the DPF does not grow with the number of mailboxes, and the client-side auditing cost is constant as well. To send a 1KB message, our client takes 20ms in C/Go and 51ms in Javascript.  
 Combined with the low client communication costs in Figures~\ref{servercommunicationgraph} and~\ref{clientcommunicationgraph}, this shows that an \name client can easily be deployed as background Javascript in a web page to create cover traffic, as explained in Section~\ref{covertraffic}.

To further explore performance implications of an \name client being embedded on a major news site, we measured the page load times of the New York Times, Washington Post, and Wall Street Journal websites. On average, these pages took 5.4, 3.4, and 2.2 seconds to load completely (over a 50MBit/sec connection), so the computation costs of our client in the browser are less than $3\%$ of current page load times and can occur in the background without impacting user experience. We also measured the sizes of the three websites (without caching) at 4.9MB, 9.1MB, and 8.2MB, respectively. Our JavaScript implementation with dependent libraries takes 72.5KB of space, so adding our code would increase a site's size by less than~$1.5\%$. 

\noindent\textbf{Auditing}. In addition to enabling improved communication efficiency, as seen above, our auditing protocol dramatically reduces computation costs for the client. Figure~\ref{auditgraph} shows the computation costs of our auditing protocol as compared to the protocol used in Riposte~\cite{riposte}, which we re-implemented for the purpose of this experiment. Unlike Riposte, where client and server computation costs for auditing are comparable, our protocol runs in $O(1)$ time on the client, taking less than \emph{5 microseconds} regardless of how many mailboxes are registered on the servers. This is about $55,000\times$ less than the client computation cost for auditing in Riposte for one million mailboxes and translates to overall client computation on our system running $8\times$ faster than it would if it were using the Riposte auditing protocol. In addition to the asymptotic improvement, our protocol uses only hardware-accelerated AES evaluations, whereas Riposte's auditing protocol involves a mix of AES evaluations and more costly SHA256 hashes. 

Our auditing protocol's performance is comparable to Riposte on the server side, but it does not require a third auditing server as Riposte does. The performance bottleneck on the servers is DPF evaluations, not auditing, so server side performance improvements in auditing would only result in negligible improvements in end-to-end performance. 
 As we will see, \name outperforms Riposte's overall throughput despite not significantly changing server side auditing costs. 
 
\begin{figure*}
\centering
\minipage{0.32\textwidth}
\begin{tikzpicture}
\begin{loglogaxis}[
    title style={align=center}, title={Message Delivery Latency},
    width=\linewidth, height=4.5cm,
    xlabel={Number of Mailboxes},
	y label style={at={(axis description cs:.1,.5)},anchor=south},
	ylabel={Time [seconds]},
    cycle list name=auto,
    ymin=0, 
    xtick={100,1000,10000,100000,1000000},
    tick pos=left,
    minor tick style={draw=none},
    legend pos=south east,
    ymajorgrids=true,
    grid style=dashed,
    	legend style={at={(1.05,-.85)}},
    legend entries={\small \name (1KB), \small Pung (1KB),\small \name (10KB), \small Pung (10KB), \small \name (32KB)},
	legend columns=2,
    ]
    \addplot[color=brown, mark=diamond*]
    coordinates {
    (100,.0600)(1000,.0737)(10000,.210)(100000,1.54)(500000,7.61)(1000000,15.1)
    };
    \addplot[color=blue, mark=*]
    coordinates {
    (100,.158)(1000,.188)(10000,.336)(100000,2)(1000000,19.3)
    };
    \addplot[color=brown, dashed, mark=diamond*]
    coordinates {
    (100,.0648)(1000,.122)(10000,.621)(100000,5.856)
    };
    \addplot[color=blue, dashed, mark=*]
    coordinates {
    (100,.128)(1000,.274)(10000,1.83)(100000,16.64)
    };
    \addplot[color=brown, dotted, mark=diamond*]
    coordinates {
    (100,.0772)(1000,.232)(10000,1.93)(100000,14.2)
    };

\end{loglogaxis}
\end{tikzpicture}
\caption{Message delivery latency in \name and Pung for various message sizes. \name outperforms Pung by $1.3-2.6\times$ for 1KB messages and by $2.0-2.9\times$ for 10KB messages. Pung's performance for 10KB messages is comparable to \name's performance for 32KB messages.}
\label{fig:latency} 
\endminipage\hfill
\minipage{0.32\textwidth}
\begin{tikzpicture}
\begin{semilogxaxis}[
    title style={align=center}, title={Message Throughput},
    width=\linewidth, height=4.5cm,
    xlabel={Number of Mailboxes},
	y label style={at={(axis description cs:.1,.5)},anchor=south},
	ylabel={Throughput [Msgs/sec]},
    cycle list name=auto,
    xtick={100,1000,10000,100000,1000000},
    tick pos=left,
    legend pos=south east,
    ymajorgrids=true,
    minor tick style={draw=none},
    grid style=dashed,
    	legend style={at={(.95,-.85)}},
		legend entries={\small Riposte with 1KB Messages, \small \name with 1KB Messages, \small \name with 32KB Messages},
    ]
    \addplot[color=red, mark=square*]
    coordinates {
    (1000,38)(10000,10)(25000,5)(50000,3.4)(75000,2.5)(100000,2)(200000,1.2)(300000,.9)(400000,.7)(500000,.6)
    };
    \addplot[color=brown, mark=diamond*]
    coordinates {
    (1000,53.54)(10000,49.35)(25000,31.41)(50000,15.81)(75000,10.6)(100000,7.87)(200000,4.05)(300000,2.81)(400000,1.92)(500000,1.57)
    };
    \addplot[color=brown, dashed, mark=diamond*]
    coordinates {
    (1000,51.49)(10000,10.6)(25000,4.41)(50000,2.24)
    };
    
        \node[] at (axis cs: 3800,.9) {\small(Higher is better)};
\end{semilogxaxis}
\end{tikzpicture}
\caption{\name's throughput is $1.4$-$6.3\times$ that of Riposte for 1KB messages. Even with 32KB messages, \name's throughput is still comparable to Riposte on 1KB messages. For large numbers of mailboxes, both systems are computation-bound by the number of DPF evaluations required to process writes.}
\label{riposte}
\endminipage\hfill
\minipage{0.32\textwidth}
\begin{tikzpicture}
\begin{semilogxaxis}[
    title style={align=center}, title={\\System Cost per 1M Messages},
    width=6cm, height=4.5cm,
    xlabel={Number of Mailboxes},
	y label style={at={(axis description cs:.1,.5)},anchor=south},
	ylabel={Cost [\$]},
    cycle list name=auto,
    ymin=.1, ymax=1100,
    xtick={100,1000,10000,100000,1000000},
    tick pos=left,
    minor tick style={draw=none},
    legend pos=south east,
    ymajorgrids=true,
    grid style=dashed,
    	legend style={at={(1.05,-.55)}},
		legend entries={\small Pung, \small Riposte, \small \name},
	legend columns=3, 
    ]
    \addplot[color=blue, mark=*]
    coordinates {
    (64,14.01)(256,14.75)(1024,17.70)(4096,67.02)(16384,106.18)(65536,318.19)(262144,1028.27)
    };
    \addplot[color=red, mark=square*]
    coordinates {
    (1000,11.35)(10000,23.15)(25000,86.3)(50000,126.80)(75000,172.6)(100000,215.73)(200000,359.29)(300000,478.95)(400000,615.83)(500000,718.51)
    };
    \addplot[color=brown, mark=diamond*]
    coordinates {
    (1000,5.39)(10000,5.84)(25000,9.14)(50000,18.12)(75000,27.00)(100000,36.35)(200000,70.59)(300000,101.70)(400000,148.82)(500000,181.99)
    };
\end{semilogxaxis} 
\end{tikzpicture}
\caption{Dollar costs to run end-to-end metadata hiding systems with cryptographic security guarantees. Prices are based on Google Cloud Platform public pricing information for compute instances and data egress. Processing one million messages in \name in the presence of 100,000 registered mailboxes costs $5.9\times$ less than the next cheapest system. }
\label{fig:costgraph}
\endminipage
\end{figure*}
 
\subsection{Server Performance}

We now measure the performance of \name on the server-side. We measure the total throughput of the system, the latency between when a client sends a message and when the mailbox owner can read it, and the cost in dollars of running \name. 

\noindent\textbf{Throughput}. We compare \name's throughput to Riposte~\cite{riposte}. Figure~\ref{riposte} shows the comparison between \name and Riposte for 1KB messages, where throughput is measured as the number of writes the servers can process per unit time. \name's throughput is $1.4$-$6.3\times$ that of Riposte in our experiments, and \name's throughput when handling 32KB messages is comparable to Riposte when handling only 1KB messages for up to about 50,000 mailboxes. Both systems are ultimately computation-bound by the number of DPF evaluations required to process writes. The graph shows the high throughput of each system \rev{drops} significantly as they shift from being communication-bound to being computation-bound by DPF evaluations for increasingly large numbers of mailboxes. 

Like \name, Riposte uses DPFs to write messages across two servers. Unlike \name, Riposte requires a third party to audit user messages and must run its protocol in rounds to provide anonymity guarantees to its users. The rounds are necessary for Riposte's anonymous broadcast setting because all messages are public, so if messages were revealed after each write, the author of a message would clearly be whoever connected to the system last. In contrast, \name messages can be delivered immediately without waiting for a round to end.

Another difference between \name and Riposte is that Riposte relies on a probabilistic approach based on hashing for users to decide where to write with their DPF queries. This means that there is a chance messages will collide when written to the same address, rendering all colliding messages unreadable. We evaluated Riposte with parameters set to allow a failure rate of $5\%$, meaning that 1 in 20 messages would be corrupted by a collision and not delivered, even after Riposte's collision-recovery procedure. \name's virtual address system avoids this issue because the space of virtual addresses is so large that collisions would only occur with negligible probability. 

\noindent\textbf{Latency}. Since \name does not require any synchronization between clients and the \name servers, the latency of a write request consists only of the time for the servers to process the request and for the mailbox owner to read the message. Figure~\ref{fig:latency} shows how latency for processing a single write request scales as the number of mailboxes increases for various mailbox sizes. After about 10,000 mailboxes, or even 1,000 mailboxes for larger message sizes, message processing becomes bound by the latency of computing AES for each DPF evaluation, so total latency increases linearly with the number of DPFs that must be evaluated (one per mailbox). 

In prior metadata-hiding communication systems, message delivery latency depends on a deployment-specified round duration. As such, it is difficult to directly compare latency in \name to prior work. We can, however, compare to the computation time on the servers to process one message and deliver it to its recipient. For example, Riposte's ``latency'' under this metric is simply the time to process a DPF write and then run an audit. A more interesting comparison is to see how \name's server-side costs compare to a different architecture, such as the single-server PIR-based approach of Pung~\cite{pung}. 

Since Pung~\cite{pung} uses fast writes and more expensive reads whereas \name has fast reads but expensive writes, we run both systems with a write followed by a read, as required by Pung's messaging use case. As shown in Figure~\ref{fig:latency}, \name outperforms Pung by $1.3$-$2.6\times$ when run with 100-1,000,000 mailboxes for 1KB messages. When we increase the message size to 10KB, we find that Pung is $2-2.9\times$ slower than \name and closely matches \name's performance on 32KB messages. Note that the comparison to Pung is not quite apples to apples because Pung operates in a stricter single-server security setting.

\noindent\textbf{Total system cost}. Having measured \name's throughput and latency, we now turn to the question of \name's cost in dollars (USD). Our evaluation focuses on the dollar cost of running the infrastructure required for \name in the cloud and excludes human costs such as paying engineers to deploy and maintain the software. The primary non-human costs in running \name, as with any metadata-hiding system, come from running the necessary servers and passing data through them. Using the data from our evaluation thus far, we estimate the price of running \name to send one million messages using public Google Cloud Platform pricing information. We calculate the cost of running the system as the cost of hosting the \name servers for the length of time required to process one million messages plus the data passed between the servers and back to the client (data passing into Google cloud instances from clients outside is free). We price the instances according to costs for various regions in the US and Canada and calculate data charges using the prices for data transfer between regions in the US and Canada (for communication between servers) or with the public internet (for communication with clients).

The results of this estimation process appear in Figure~\ref{fig:costgraph}, where we carry out similar calculations for Pung and Riposte. As depicted in the figure, processing one million messages with \name costs $5.9\times$ less than Riposte, the closest prior work measured, in the presence of 100,000 mailboxes. The high cost of running Pung comes from its communication costs, where data egress charges far outweigh the cost of hosting the system. \rev{The data egress cost of sending one million messages in Pung with 262,144 registered mailboxes exceeds $\$1,000$}. On the other hand, \name and Riposte incur smaller data costs, $\$0.05$ per million messages in \name and $\$4.21$ per million messages in Riposte with one million registered mailboxes. The large gap in cost between \name and Riposte comes from hosting the servers themselves. \name's higher throughput means it can process one million messages more quickly than Riposte, and the fact that it requires only two servers, compared to three in Riposte, means that the cost per hour of running \name is approximately $2/3$ that of running Riposte. \rev{Hosting costs per 24 hours, excluding data costs, are $\$11.75$ for Pung, $\$37.25$ for Riposte, and $\$24.68$ for \name, corresponding to the number of servers each system needs (including cost differences for hosting servers in different regions).}

\noindent\textbf{\rev{Comparison to differential privacy systems.}} \rev{As described in Section~\ref{designapproaches}, systems based on differential privacy (DP) exchange gradual metadata leakage over time for stronger performance. Although this fundamental difference in security properties makes it difficult to do a direct comparison to DP systems such as Vuvuzela~\cite{vuvuzela}, Stadium~\cite{stadium}, and Karaoke~\cite{karaoke}, we will attempt here to roughly compare \name to published performance results for Vuvuzela and Karaoke. Vuvuzela operates with the same distributed trust model as \name, with a small number of servers, whereas Karaoke is designed for use in a setting with many servers. See Section~\ref{designapproaches} for a discussion of these two approaches to distributing trust. }

\rev{One further difference to keep in mind when comparing existing DP systems to \name (as well as the systems we have compared \name to thus far) is that costs in Riposte, Pung, and \name increase in the number of mailboxes registered, while costs in existing DP-based systems increase in the number of \emph{users} registered. This means that a fully connected communication graph on $N$ users would require $N^2$ mailboxes in \name but would not require additional cost in DP systems beyond that of $N$ users and the high volume of traffic required for all of them to talk to each other. Fortunately, in most messaging systems, each user only has a small number of active contacts relative to the total number of users on the platform, so this difference should not cause harm in practice.}

Vuvuzela’s end-to-end latency to deliver a 256 byte message for the lowest security setting on which it was evaluated hovers around 8 seconds for 10,000 users and 20 seconds for one million users~\cite{vuvuzela}. By comparison, Express takes 210ms to write and then read a larger 1KB message when there are 10,000 mailboxes and 15 seconds when there are one million mailboxes. The higher latency in Vuvuzela is due to cover traffic messages sent before a message can be delivered.

\rev{Karaoke operates using a variable number of servers, and its end-to-end latency to deliver a 256 byte message  hovers around 6 seconds for one million users and 100 servers when up to $20\%$ of servers are malicious~\cite{karaoke}. However, Karaoke's latency approximately triples when moving from providing security against $20\%$ malicious servers to $50\%$ malicious servers, which more closely matches the one-out-of-two security provided by \name. Since Karaoke's evaluation was also conducted on more powerful machines than ours, we conclude that latency is roughly comparable between \name and Karaoke.}


\rev{On the other hand,} not requiring cryptographic security allows \rev{DP solutions} to achieve higher throughput than cryptographic systems. As such, \rev{they} can process messages faster and at lower cost than \name. However, in addition to the difference in security guarantees, \rev{they} achieve their low price by pushing the true cost of operating the system onto clients. To send and receive messages, clients must \emph{always} remain online. 

\section{Related Work}\label{related}
  
The most widely used anonymity system in use today is without a doubt Tor~\cite{tor}, which relies on onion routing. SecureDrop~\cite{securedrop,securedropreport} is a widely used Tor-based tool to allow sources to anonymously connect with journalists to give tips. Although our work focuses on hiding metadata and not on preserving anonymity, anonymity systems are often used even when clients only wish to hide metadata. Although a number of works precisely model and analyze the security offered by Tor~\cite{BKM+16,KBS+19,KBS20}, it is unfortunately vulnerable to traffic analysis attacks if a passive adversary controls enough of the  network~\cite{breaktor,HB13,JWJ+13}. A recent impossibility result suggests that this limitation may be necessary for broad classes of anonymity systems~\cite{trilemma}.

\noindent\textbf{Cryptographic security}. \name belongs to a broad family of works which aim to give cryptographic guarantees regarding anonymity and metadata-hiding properties. One category of works in this area include systems based on mix-nets~\cite{verdict,dissent,herbivore,riffle,pynchon,dissentnumbers,loopix} which involve all users in a peer to peer system participating in shuffling messages~\cite{Chaum81,Chaum88}. Later work has added verifiability to this model~\cite{riffle} and outsourced the shuffling to a smaller set of servers~\cite{pynchon,dissentnumbers}. Most recently, mixing techniques have been extended to support large numbers of users in Atom~\cite{atom} and XRD~\cite{xrd}. Systems in this line of work suffer from high latency due to the need to run many shuffles and require participation by a large number of servers run by different operators to achieve security.

An important difference between \name and mixnets relates to tradeoffs in anonymity and latency. 
Since a user's anonymity set is based on the number of messages being shuffled together, a mixnet operator must choose between a high-latency setting with a large anonymity set or a lower latency setting with a smaller anonymity set. For example, if 1,000 messages are sent through the system in one hour, a mixnet that wants an anonymity set size of 1,000 must wait an hour before it can deliver messages, whereas \name can achieve the same anonymity set but deliver messages immediately. A mixnet's anonymity set is restricted to the number of messages included in the mixing, which in turn depends on the desired latency, leading to an inherent tradeoff between anonymity and latency~\cite{trilemma}. Express messages, on the other hand, are in some sense mixed with all the prior messages sent through the system. This means that while a mixnet may have to compromise on anonymity set size to meet a given latency goal, \name does not. 

Another class of cryptographic messaging solutions use private information retrieval techniques~\cite{pir,xpir,dpfs,dpfs2,OS97,ACLS18} to render reads or writes into a database of mailboxes private and target a variety of use cases~\cite{riposte,talek,P3,pung,dp5,CB95,pushpull}. \name falls into this category. Riposte~\cite{riposte} \rev{and, more recently, Blinder~\cite{blinder}}, provide anonymous broadcast mechanisms using DPFs~\cite{dpfs}, and Talek~\cite{talek} offers a private publish-subscribe protocol. P3~\cite{P3} deals with privately retrieving messages with more expressive search queries. Pung~\cite{pung} operates in a single-server setting and therefore requires weaker trust assumptions than \name, but as we show in Section~\ref{eval}, has higher costs than \name as well.

\noindent\textbf{Differential privacy}. 
Another class of works make differential privacy guarantees~\cite{Dwork06} instead of cryptographic guarantees. These systems typically achieve better performance but at the cost of setting a \emph{privacy budget} that dictates how much privacy the system will provide. These works include Vuvuzela~\cite{vuvuzela}, Alpenhorn~\cite{alpenhorn}, Stadium~\cite{stadium}, and Karaoke~\cite{karaoke}.

\section{Conclusion}
 
We have presented \name, a metadata-hiding communication system that requires only symmetric key cryptographic primitives while providing near-optimal communication costs. In addition to order of magnitude improvements in communication cost, \name reduces the dollar cost of running a metadata-hiding communication system by $6\times$ compared to prior work. 
Our implementation is open source and available online at \url{https://github.com/SabaEskandarian/Express}. 

\section*{Acknowledgments}
  
  We would like to thank Dima Kogan, Alex Ozdemir, the anonymous reviewers, and our shepherd, Esfandiar Mohammadi, for their thoughtful comments. 
  
  This research was supported in part by affiliate members and other supporters of the Stanford DAWN project---Ant Financial, Facebook, Google, Infosys, NEC, and VMware---as well the NSF under CAREER grant CNS-1651570. The work was additionally funded by NSF, DARPA, a grant from ONR, and the Simons Foundation. Any opinions, findings, and conclusions or recommendations expressed in this material are those of the authors and do not necessarily reflect the views of DARPA or the National Science Foundation.
  
{\small
\setstretch{0.95}
\bibliography{main}
\bibliographystyle{plain}
}

\appendix
\section{Security Arguments}
\label{app:security}

This appendix formalizes and proves the soundness and metadata-hiding security properties described in Section~\ref{defs}. 

\noindent\textbf{Soundness}. We formalize soundness as follows. 

\begin{definition}[Soundness]
We define the following soundness game \textsf{SOUND[$\lambda$]} played between an adversary $\mathcal{A}$ and a challenger $\mathcal{C}$ who simulates the behavior of servers $A$ and $B$. Both $\mathcal{A}$ and $\cal{C}$ are given $\lambda$ as input. 

\begin{denseitemize}
\item Setup. Challenger $\mathcal{C}$ creates an initially empty list $I$ of compromised mailbox \rev{indices}. Adversary $\mathcal{A}$ requests creation of a number of mailboxes $N$ of its choosing. There are two ways in which it may create a mailbox:
\begin{denseenumerate}
\item Adversary $\mathcal{A}$ performs the role of a user interacting with the servers to create a new mailbox. Challenger $\mathcal{C}$ adds this mailbox to $I$. 

\item Adversary $\mathcal{A}$ instructs $\mathcal{C}$ to create a mailbox where $\mathcal{C}$ plays the role of both the user and the servers, saving the user's state (and in particular, the mailbox keys) at the end of the registration process. 
\end{denseenumerate}

\item Queries and Corruptions. Adversary $\mathcal{A}$ sends requests to the servers, controlled by $\mathcal{C}$. At any time, it may send $\mathcal{C}$ a mailbox index $i$, at which point $\mathcal{C}$ will send the saved state of the user who registered mailbox $i$ and add $i$ to list $I$. 

\item Output. Challenger $\mathcal{C}$ performs a read on each registered mailbox. If $|I|<N$ and any mailbox outside of the list $I$ contains nonzero contents, the adversary wins the game.
\end{denseitemize}
 
We say a messaging scheme is \emph{sound} if no PPT adversary can win the soundness game above with greater than negligible probability in the security parameter $\lambda$.
\end{definition}

\begin{claim}
The \name scheme is sound. 
\end{claim}

\begin{proof}
The soundness proof follows closely from the soundness of our auditing protocol. 
For each write request sent to the \name servers, we consider two cases: where the write modifies one mailbox and where the write modifies more than one mailbox. If a write modifies more than one mailbox, then it will not be applied to the database of mailboxes, except with negligible probability in $\lambda$, by the soundness property of the auditing protocol. This means that we must only consider writes that modify a single mailbox. The adversary does not know the virtual addresses of mailboxes outside of $I$, but it only wins the soundness game if it produces a DPF that writes to the address of a mailbox outside of $I$. This can only occur with probability $2^{-\lambda}$ (for $\lambda = 128$ in our instantiation of the protocol), which is also negligible. Thus an adversary can only win the soundness game with probability negligible in $\lambda$. 
\end{proof}

\noindent\textbf{Metadata-hiding}. We can formalize the definition of metadata-hiding by requiring that there exists an efficient \emph{simulator} algorithm \textsf{Sim} that, given the list $\ell$ of honest clients who connect with the servers, produces an output which is computationally indistinguishable from the view of an adversary $\mathcal{A}$ who controls any number of users and one server while processing requests from the remaining honest users, subject to the restriction that the recipients of the messages from honest users are never among those controlled by~$\mathcal{A}$. More specifically, $\ell$ should include which client connects, time of connection, and size of message transmitted for each connection made to the compromised server. Given this information, the client can simulate the content of the messages sent by the honest client. 

This definition satisfies our intuitive notion of metadata-hiding because it means that for each message, the server learns nothing about who the message is sent to, as everything it learns could be simulated before it even sees the request. This information would be contained in the content of the honest client's messages, which are not given to the simulator. We sketch a proof of the metadata-hiding security argument below. The proof relies on the zero-knowledge property of the auditing protocol, the privacy of the DPFs used, and the security of the encryption used for access control. 

\begin{claim}[Informal]
There exists an algorithm \textsf{Sim} that, given the list $\ell$ of honest client connections to the \name servers, simulates the view of an adversary $\mathcal{A}$ who controls one \name server and any number of clients, subject to the restriction that the recipients of the honest clients' messages are never among those controlled by~$\mathcal{A}$. 
\end{claim}

\begin{proof}[Proof (sketch).]

\textsf{Sim} simulates write requests from honest users and the process of auditing them by invoking the simulator implied by the zero-knowledge property of the auditing protocol. Note that this in turn uses the simulator implied by the definition of DPF privacy to generate DPF shares. Moreover, whenever malicious users request to read the contents of mailboxes, the simulated honest server(s) returns encryptions of zero. 

The proof that this simulator gives the adversary $\mathcal{A}$ a view indistinguishable from interaction with a real honest server and honest users is fairly straightforward. First, since the adversary knows the virtual addresses of honest users' mailboxes, as well as one of the two keys needed to read the contents of those mailboxes (if it has compromised one of the servers), it can send read requests for the contents of honest mailboxes. However, since the adversary does not see the second key to any honest users' mailboxes, we invoke the semantic security of the encryption scheme used to protect honest mailbox contents to show that the messages returned from read requests to an honest server are indistinguishable from encryptions of zero.

From here, just as in the case of soundness, the proof follows from the security of the auditing scheme. From the zero-knowledge property of the auditing scheme, we know that the view of either server in the auditing protocol can be simulated. But the view of each server in \name's auditing protocol is the same as the view of that server in the overall protocol, since the server's view only consists of its shares of the proof input (in the compressed form of a DPF share from which it derives the actual inputs) and the proof messages themselves. 
\end{proof}

\section{SNIPs and Analysis of Auditing Protocol}
\label{app:snips}

This appendix sketches the instantiation of the proofs used in our auditing protocol as well as the analysis of the auditing protocol. Full details and a security proof for this proof system can be found in the Prio paper~\cite{prio}. We include the instantiation of the proof here for completeness, including some improvements described in the follow-up work of Boneh et al.~\cite{BBC+19}. 

The size of a SNIP proof is linear in the number of multiplication gates in the arithmetic circuit representing the statement to be proved. In our case, there are 2 multiplications. The client numbers the gates as $1$ and $2$. The idea of the proof is to create three polynomials $f$, $g$, and $h$ such that $f, g$ represent the left and right inputs of each gate and $h$ to the outputs of each gate. $f$ is the polynomial defined by the points $(0, r_f), (1, c), (2, m)$, and $g$ is the polynomial defined by the points $(0, r_g), (1, c), (2, C)$, where $r_f$ and $r_g$ are random values chosen by the client. Observe that the servers already hold shares of each point used to define $f$ and $g$
except the random values $r_f$ and $r_g$, shares of which must be included in the SNIP proof.

Next, $h$ is defined as the polynomial representing the expected outputs of each multiplication gate, or the product $f\cdot g$. Since each of $f$ and $g$ will be of degree $2$, $h$ will be of degree $4$. The client can compute $h$ from $f$ and $g$ and must send shares of the description of $h$ to each server as part of the proof. 

Since the servers now have shares of the inputs and outputs of each multiplication from $f$, $g$, and $h$, they only need to check that $f\cdot g = h$ to be convinced that this relationship holds among their inputs. They do this by evaluating each polynomial at a random point $t$ and checking equality. To compute the product $f(t)\cdot g(t)$, the servers simply evaluate their shares of each function and publish the result. This reveals nothing about $f$ or $g$ except their value at the point $t$. 

The Prio paper~\cite{prio} and the improvements of Boneh et al.~\cite{BBC+19} give full proofs of completeness, soundness, and zero-knowledge for this protocol. As a minor optimization, instead of sending one proof as described above, we send two separate SNIPs, one for each of the two multiplications. This results in a slightly larger proof size but simplifies the polynomial multiplications because the polynomials $f$, $g$ become linear and $h$ becomes quadratic. The security properties of the protocol are unchanged by this modification. 

\noindent\textbf{Analysis.} Having described the relevant building blocks, we now sketch the analysis of our full auditing protocol. The security properties of our auditing scheme follow directly from those of the two protocols we combine to build it (which we do not re-prove here). Completeness follows directly from the completeness of the verifiable DPF protocol of Boyle et al. as well as the completeness of SNIPs. 

Likewise, soundness follows directly from the soundness of these two building blocks, with soundness error equal to the sum of the soundness error of the DPF verification protocol and the SNIP. We prove the following claim. 

\begin{claim}
If the servers begin the auditing protocol holding
vectors $w_A \in \F^n$ and $w_B \in \F^n$ such that
$w = w_A + w_B \in \F^n$ is a vector of Hamming-weight 
\emph{greater than} one, then the audit will reject, except with 
error probability $\epsilon = O(1/|\F|)$.
\end{claim}

By taking $\F$ to be a field of size $2^\lambda$, for security parameter
$\lambda$, we can make the error probability $\epsilon$ negligibly small
in $\lambda$.

The claim is true because the auditing protocol will only accept a false proof if (1) the difference $c^2-mC=0$ for a $w$ that has more than one non-zero entry, or (2) the soundness of the SNIP fails to enforce that only inputs satisfying this relationship will be accepted.
But the probability of (1) is negligible in $|\F|$ by the security of the DPF verification protocol of Boyle et al.~\cite{dpfs2}, and the probability of (2) is negligible in $|\F|$ by the soundness of SNIPs~\cite{prio,BBC+19}.
By a union bound, the soundness error of the overall protocol is at most the sum of the soundness errors of the verifiable DPF protocol and the SNIPs. 

To prove the zero-knowledge property, we must show that there exists a simulator algorithm $\mathsf{Sim}$ that can produce outputs whose distribution is computationally indistinguishable from the view of the servers in an execution of the \name auditing protocol where the sum $w_A+w_B$ corresponds to a vector with a single non-zero entry. This algorithm will interact with a potentially malicious adversary $\mathcal{A}$ who plays the role of the server whose view is being simulated. This proves the security of the protocol because it shows that an adversary can learn anything it would learn from actually participating in the protocol by running $\mathsf{Sim}$ on its own. 

The construction of $\mathsf{Sim}$ and subsequent proof of security follow almost directly from the original proof of security for SNIPs used in Prio~\cite{prio}. To see why, observe that the view of each server in the auditing protocol consists of the server's DPF share, the server's share of the proof, and any messages sent between the servers during the proof. The only difference between this and the standard SNIP simulator is that the server's inputs are compressed in the form of DPF shares instead of being stated explicitly as the vector $w_A$ or $w_B$. In essence, the DPF can be thought of as an efficient way to encode the server's inputs to the proof. To bridge this difference between our protocol and the original SNIP, we make one small change to the SNIP simulator. The original SNIP simulator samples the server's input share at random. Our modified SNIP simulator will sample the server's input shares using simulated DPF shares instead. Since the proof of zero-knowledge is otherwise identical, we defer to the prio paper for the full proof~\cite{prio}. 

\end{document}